\documentclass[runningheads]{llncs}
\usepackage{graphicx}
\usepackage{amsmath}
\usepackage{amssymb}
\usepackage{tikz}
\usepackage{stmaryrd}
\usepackage[colorlinks=true,breaklinks=true,linkcolor=blue]{hyperref}
\usepackage{color}
\usetikzlibrary{arrows,arrows.meta,decorations.pathmorphing,backgrounds,positioning,fit,petri,knots,patterns}

\newcommand{\A}{\mathcal{A}}
\newcommand{\C}{\mathcal{C}}
\newcommand{\Z}{\mathbb{Z}}
\newcommand{\N}{\mathbb{N}}
\newcommand{\Le}{\mathcal{L}}
\newcommand{\F}{\mathcal{F}}

\newcommand{\M}{\mathcal{M}}
\newcommand{\lc}{\mathcal{L}}

\newcommand{\iC}[1]{\bigcap_{P \in {#1}^{n^2}}S_{a_P}(P)}
\newcommand{\Pe}{\mathcal{P}}
\newcommand{\ob}[1]{\bar{#1}}

\newcommand{\bsq}{\blacksquare}

\newcommand{\nsq}{\llbracket 0, n \rrbracket^2}
\DeclareMathOperator{\supp}{supp}
\DeclareMathOperator{\mine}{min}
\DeclareMathOperator{\east}{East}
\DeclareMathOperator{\west}{West}
\DeclareMathOperator{\north}{North}
\DeclareMathOperator{\south}{South}
\newcommand{\ed}{=}
\newcommand{\Sb}[1]{\ob{\Sigma}_{#1}}
\newcommand{\Pb}[1]{\ob{\Pi}_{#1}}

\newcommand{\grey}{black!50!white}

\newcommand{\coTOTAL}{\mathrm{coTOTAL}}
\newcommand{\MFOSUB}{\mathrm{FOSUB}}

\newcommand{\stateb}[3]{
    \fill[white] (#1,#2) circle (0.4);
    \draw (#1,#2) node {$q_{#3}$};
}

\title{Multidimensional tilings and MSO logic}

\author{Rémi PALLEN\inst{1}\orcidID{0009-0006-6708-6904} \and Ilkka TÖRMÄ\inst{2}\orcidID{0000-0001-5541-8517}}

\institute{ENS Paris-Saclay, Gif-sur-Yvette, France \email{remi.pallen@loria.fr} \and Department of Mathematics and Statistics, University of Turku, Turku, Finland \email{iatorm@utu.fi}}

\begin{document}

\maketitle
\begin{abstract}
We define sets of coulourings of the infinite discrete plane using monadic second order (MSO) formulas. We determine the complexity of deciding whether such a formula defines a subshift, parametrized on the quantifier alternation complexity of the formula. We also study the complexities of languages of MSO-definable sets, giving either an exact classification or upper and lower bounds for each quantifier alternation class.

\keywords{MSO logic \and subshifts \and symbolic dynamics \and tilings}
\end{abstract}

\section{Introduction}

A \emph{tiling} or \emph{configuration} is a colouring of the two-dimensional plane $\Z^2$ with finitely many colours. A \emph{subshift} is the set of tilings such that no pattern from a set of \emph{forbidden patterns} appears. These notions were first introduced by Wang \cite{Wang} to study first order logic, and to find an algorithm which computes whether a formula is a tautology or not. Berger proved in 1966 \cite{Berger} the undecidability of the domino problem, and many other problems were proved undecidable too thanks to this result in the following decades.

This article continues the line of research in \cite{j-t,torma}. The idea is to define sets of configurations using \emph{monadic second order} (\emph{MSO} for short) logic.
MSO definability is a classical and broad area of research; see \cite{Wo97} for an overview in the context of finite and infinite languages.
MSO formulas on finite words correspond to regular languages, and those on infinite words correspond to $\omega$-regular languages.
Subshifts can be seen as languages of two-dimensional infinite words, so it makes sense to ask which of them can be defined by MSO formulas.
In the formalism of \cite{j-t}, an MSO formula always defines a translation invariant set, but if it has a first order existential quantifier, it may not define a topologically closed set, i.e.\ a subshift. This is why \cite{j-t,torma} mostly restrict their study to the case where the MSO formulas do not have first order existential quantifiers. In this article, we study the case where they are allowed.

Our first result is to determine the exact complexity of deciding whether an MSO formula defines a subshift or not, depending on the complexity (number of second-order quantifier alternations) of the formula. We also study the complexity of languages of subshifts and sets these formulas can define.

One motivation for this research is \emph{Griddy} (formerly called Diddy), a Python library for discrete dynamical systems research developed by Ville Salo and the second author \cite{diddyrepo,diddypaper}.
Griddy allows the definition and manipulation of multidimensional subshifts defined by first order logical formulas.
Even though most properties of multidimensional subshifts are undecidable, partial algorithms can often give good solutions or approximations in ``naturally occurring'' cases.
This raises the question whether the same holds for subshifts defined by MSO formulas, or whether the complexity is too high to obtain even partial algorithms.

\section{Preliminaries}
\subsection{Subshifts}

We review some basic facts from multidimensional symbolic dynamics.
See \cite{LiMa21}, especially Section A.6, for a more comprehensive overview.

Let $\A$ be a finite set of colours called \emph{alphabet}. A \emph{pattern} is a map $P:D \rightarrow \A$ where $D \subset \Z^2$. $D$ is called \emph{support} of the pattern and is denoted $D(P)$. A \emph{configuration} is a pattern $x$ such that $D(x)=\Z^2$. We say that a pattern $P$ \emph{appears} in a configuration $x \in \A^{\Z^2}$ if there exists $\vec{v} \in \Z^2$ such that $P_{\vec{u}}=x_{\vec{u}+\vec{v}}$ holds for all $\vec{u} \in \supp(P)$. The number of occurrences of $P$ in $x$ (which may be infinite) is denoted $\#_x(P)$. A \emph{cell} or \emph{position} is an element $\vec{v} \in \Z^2$.

For a set of patterns $\F$, we define $X_{\F} \subseteq \A^{\Z^2}$ as the set of those $x \in \A^{\Z^d}$ in which no pattern from $\F$ appears. Such a set is called a \emph{subshift}. When $\F$ is finite, $X_{\F}$ is a \emph{subshift of finite type (SFT)}.

\begin{property}\label{computationtiles}
    For each Turing machine $M$, there exists a SFT (see Figure \ref{fig:computationtiles} for the colours, a pattern of size $1 \times 2$ or $2 \times 1$ is forbidden if the tiles do not match), such that the \emph{initial tile} may appear in a tiling if and only if $M$ does not halt on the empty input. There exists a slightly different version for which an input is written on the computation tape. In this case, if the initial tile appears on a configuration with a finite input $x$ written on it, $M$ does not halt on $x$.
\end{property}

\begin{figure}[h]
    \centering
\begin{tikzpicture}[scale=1.5]
\fill[gray!40] (0,0) -- (1,0) -- (1,1) -- (0,1) -- cycle;

\draw (1.5,0) -- (1.5,1) -- (2.5,1) -- (2.5,0) -- cycle;
\fill[gray!40] (1.5,0) -- (1.5,1)  -- (1.75,1) -- (1.75,0) -- cycle;

\draw (3,0) -- (3,1) -- (4,1) -- (4,0) -- cycle;
\fill[gray!40] (3,0) -- (3,0.25)  -- (4,0.25) -- (4,0) -- cycle;
\draw (3,0.5) -- (3.25,0.5) -- (3.25,1);
\draw (3.5,1) node[above] {\tiny $B$};

\draw (0,-1.5) -- (0,-0.5) -- (1,-0.5) -- (1,-1.5) -- cycle;
\fill[gray!40] (0,-1.5) -- (0,-1.25)  -- (1,-1.25) -- (1,-1.5) -- cycle;
\draw (0,-1) -- (0.25,-1) -- (0.25,-0.5);
\draw (0,-1) node[left] {\tiny $q_0$};
\draw (0.25,-0.5) node[above] {\tiny $q_0,B$};

\draw (3,-1.5) -- (4,-1.5) -- (4,-0.5) -- (3,-0.5) -- cycle;
\draw (3.4,-1.1) -- (3.4,-0.9) -- (3.6,-0.9) -- (3.6,-1.1) -- cycle;
\draw (3.5,-1) node {$a$};
\draw (3.5,-0.5) node[above] {\tiny $a$};
\draw (3.5,-1.5) node[below] {\tiny $a$};

\draw (1.5,-1.5) -- (2.5,-1.5) -- (2.5,-0.5) -- (1.5,-0.5) -- cycle;
\draw (1.9,-1.1) -- (1.9,-0.9) -- (2.1,-0.9) -- (2.1,-1.1) -- cycle;
\draw (2,-1) node {$a$};
\draw (2.5,-0.75) -- (1.75,-0.75) -- (1.75,-0.5);
\draw (1.75,-0.5) node[above] {\tiny $q,a$};
\draw (2.5,-0.75) node[right] {\tiny $q,\leftarrow$};
\draw (2,-1.5) node[below] {\tiny $a$};

\draw (4.5,0) -- (5.5,0) -- (5.5,1) -- (4.5,1) -- cycle;
\draw (4.9,0.4) -- (4.9,0.6) -- (5.1,0.6) -- (5.1,0.4) -- cycle;
\draw (5,0.5) node {$a$};
\draw (4.5,0.75) -- (4.75,0.75) -- (4.75,1);
\draw (4.75,1) node[above] {\tiny $q,a$};
\draw (4.5,0.75) node[left] {\tiny $q,\rightarrow$};
\draw (5,0) node[below] {\tiny $a$};

\draw (6,-1.5) -- (7,-1.5) -- (7,-0.5) -- (6,-0.5) -- cycle;
\draw (6.4,-1.1) -- (6.4,-0.9) -- (6.6,-0.9) -- (6.6,-1.1) -- cycle;
\draw (6.5,-1) node {$a$};
\draw (6.25,-1) circle (0.1);
\draw (6.25,-1) node {\small $q$};
\draw[->] (6.25,-1.5) -- (6.25,-1.1);
\draw (6.25,-0.9) -- (6.25,-0.75) -- (7,-0.75);
\draw (6.5,-0.5) node[above] {\tiny $a'$};
\draw (6.25,-1.5) node[below] {\tiny $q,a$};
\draw (7,-0.75) node[right] {\tiny $q',\rightarrow$};

\draw (6.7,-1.7) node {\tiny If $\delta(q,a)=(q',a',\rightarrow)$};

\draw (6,0) -- (6,1) -- (7,1) -- (7,0) -- cycle;
\fill[gray!40] (6,0) -- (6,1) -- (6.25,1) -- (6.25,0.25) -- (7,0.25) -- (7,0) -- cycle;
\draw (6.5,0.5) -- (7,0.5);
\draw (7,0.5) node[right] {\tiny $q_0$};
\draw (6.5,0) node[below] {\tiny \textit{Initial tile}};

\draw (4.5,-1.5) -- (5.5,-1.5) -- (5.5,-0.5) -- (4.5,-0.5) -- cycle;
\draw (4.9,-1.1) -- (4.9,-0.9) -- (5.1,-0.9) -- (5.1,-1.1) -- cycle;
\draw (5,-1) node {$a$};
\draw (4.75,-1) circle (0.1);
\draw (4.75,-1) node {\small $q$};
\draw[->] (4.75,-1.5) -- (4.75,-1.1);
\draw (4.75,-0.9) -- (4.75,-0.75) -- (4.5,-0.75);
\draw (5,-0.5) node[above] {\tiny $a'$};
\draw (4.75,-1.5) node[below] {\tiny $q,a$};
\draw (4.5,-0.75) node[left] {\tiny $q',\leftarrow$};

\draw (4.8,-1.7) node {\tiny If $\delta(q,a)=(q',a',\leftarrow)$};
\draw (5.75,-1.9) node {\tiny for $a \in \Sigma, \ q \in Q$};
\end{tikzpicture}
    \caption{Computation tiles}
    \label{fig:computationtiles}
\end{figure}
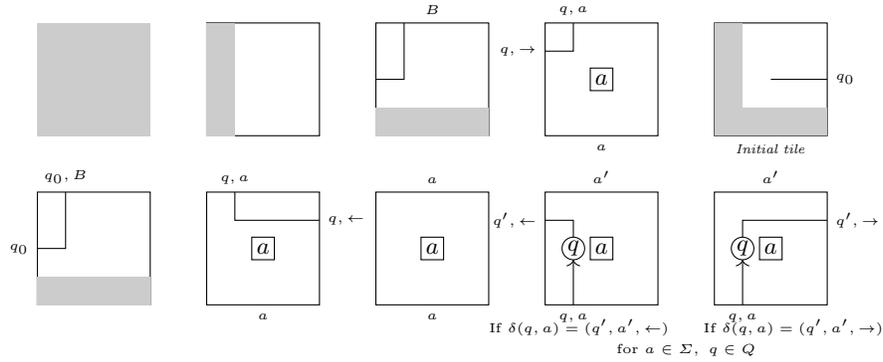

For $\vec{v} \in \Z^2$ and $x \in \A^{\Z^2}$, denote by $\sigma_{\vec{v}}(x)$ the configuration $\sigma_{\vec{v}}(x)_{\vec{u}}=x_{\vec{u}+\vec{v}}$ for all $\vec{u} \in \Z^2$. The function $\sigma_{\vec{v}} : \A^{\Z^2} \to \A^{\Z^2}$ is called the \emph{shift by $\vec{v}$}. A set of configurations $E$ is \emph{shift-invariant} if $\sigma_{\vec{v}}(x) \in E$ for all $x \in E$ and $\vec{v} \in \Z^2$.

We give a distance function $d$ on the space $\A^{\Z^2}$. For two configurations $x \neq y$, let $d(x,y)=2^{-\mine\{|a|+|b| \mid x_{(a,b)} \neq y_{(a,b)}\}}$ and $d(x,x)=0$ for all $x$. The set $\A^{\Z^2}$ is therefore a metric space and has a corresponding topology, which is generated by the clopen \emph{cylinder sets}
$[P] = \{ x \in \A^{\Z^2} \mid \forall \vec{v} \in \supp(P) \ x_{\vec{v}} = P_{\vec{v}} \}$
for finite patterns $P$.

\begin{property}
A set of configurations $X$ is a subshift if and only if it is shift-invariant and topologically closed.
\end{property}

The \emph{language} $\lc(X)$ of a set of configurations $X$ is the set of finite patterns $P$ such that $P$ appears in at least one $x \in X$.

\begin{property}
If $X$ and $Y$ are subshifts such that $\lc(X)=\lc(Y)$, then $X=Y$.
\end{property}

\subsection{Computability}
A decision problem is a predicate with some free variables (first or second order). The \emph{arithmetic hierarchy} for decision problems is defined inductively as follows:
\begin{itemize}
	\item $\Sigma_0^0$ and $\Pi^0_0$ are the set of decidable problems.
	\item $R$ is in $\Pi_{n+1}^0$ if there exists $S\in \Sigma_n^0$ such that $R \Leftrightarrow \forall x \ S(x)$.
    \item $R$ is in $\Sigma_{n+1}^0$ if there exists $S\in \Pi_n^0$ such that $R\Leftrightarrow \exists x \ S(x)$.
\end{itemize}

Here, quantifiers range over natural numbers. We define in a simililar way the analytical hierarchy:

\begin{itemize}
	\item $\Sigma_0^1$ and $\Pi^1_0$ are the set of problems in the arithmetical hierarchy.
	\item $R$ is in $\Pi_{n+1}^1$ if there exists $S\in \Sigma_n^1$ such that $R\Leftrightarrow \forall_1 X \ S(X)$.
    \item $R$ is in $\Sigma_{n+1}^1$ if there exists $S\in \Pi_n^1$ such that $R(x)\Leftrightarrow \exists_1 X \ S(X)$.
\end{itemize}
Here, $\forall_1$ and $\exists_1$ quantify over sets of natural numbers. We often write simply $\exists$ and $\forall$ with a capital letter variable when the context is clear.

We use \emph{many-one reductions} to compare the computational complexity of decision problems. We denote $Q \leq P$ is there exists a computable function $f:\N \rightarrow \N$ such that $Q(x) \Leftrightarrow P(f(x))$ for all inputs $x$. For a class $C$ and a problem $P$, if $Q \le P$ holds for all $Q \in C$, then $P$ is $C$-\emph{hard}. If $P$ is in $C$ and $C$-hard, we say that it is $C$-\emph{complete}.

\subsection{MSO Logic}

We now review the formalism used in \cite{j-t,torma} for defining subshifts (and, more generally, shift-invariant sets of configurations) using logical formulas.
Let $\A$ be a finite set of colours.
A \emph{term} is either a first order variable or one of $\mathrm{East}(t)$, $\mathrm{West}(t)$, $\mathrm{North}(t)$ and $\mathrm{South}(t)$, where $t$ is a term.
An \emph{atomic formula} is $t = t'$ or $P_c(t)$, where $t$ and $t'$ are terms and $c \in \A$ is a colour.
A \emph{monadic second order} (\emph{MSO} for short) formula is constructed from terms using the logical connectives ${\neg}$, ${\wedge}$ and ${\vee}$, plus first and second order quantification.
In general, we will use uppercase letters for second order variables and lowercase letters for first order variables.
Second order quantification is sometimes denoted $\forall_1$ and $\exists_1$ for clarity.
The first order variables will range over $\Z^2$ and the second order variables over subsets of $\Z^2$.
The \emph{radius} of a formula is the maximal depth of terms $t$ in it. A formula is \emph{closed} if all variables inside the formula are bound by a quantifier. It is \emph{first order} (\emph{FO} for short) if it contains no second order quantifier.

\begin{remark}
We concentrate on MSO formulas of the form $Q_1 X_1 Q_2 X_2 \cdots Q_n X_n \psi$,
where the $Q_i$ are second order quantifiers and $\psi$ is first order, in order to discuss its number of second order quantifier alternations.
A formula with mixed quantifier order can always be put into this form by replacing each out-of-order first order variable with a second order variable constrained to be a singleton. 
For example, $\forall x \ \exists Y \ \psi$ with $\psi$ first order becomes $\forall X \ \exists Y \ \lnot (\exists n \ \forall m \ m \in X \Leftrightarrow m=n) \lor (\exists x \ x \in X \land \psi)$.
\end{remark}

We define a hierarchy called \emph{MSO hierarchy} of the complexity of formulas:
\begin{itemize}
    \item $\Sb{0}$ and $\Pb{0}$ are the set of first order formulas.
    
    \item $\phi \in \Sb{n+1}$ if it is of the form $\exists X_1 \ \hdots \exists X_n \ \psi$ with $\psi \in \Pb{n}$.

    \item $\phi \in \Pb{n+1}$ if it is of the form $\forall X_1 \ \hdots \forall X_n \ \psi$ with $\psi \in \Sb{n}$.
\end{itemize}

For any configuration $x \in \A^{\Z^2}$, we consider the model $\M_x=(\Z^2, \tau)$, whose domain is $\Z^2$ and whose signature $\tau$ contains
\begin{itemize}
    \item four unary functions $\east$, $\west$, $\north$ and $\south$, whose interpretations are $\east^{\M_x}(a,b)=(a+1,b)$,$\west^{\M_x}(a,b)=(a-1,b)$, $\north^{\M_x}(a,b)=(a,b+1)$ and $\south^{\M_x}(a,b)=(a,b-1)$, and
    \item a unary predicate symbol $P_c$ for each $c \in \A$, with the interpretation $P_c^{\M_x}(\vec{v})=\top$ if and only if $x_{\vec{v}}=c$.
\end{itemize}
For a formula $\phi$, we write $x \vDash \phi$ if $\M_x \vDash \phi$, that is, $\phi$ holds in the model $\M_x$.
For each MSO formula $\phi$, denote $X_{\phi}=\{x \in \A^{\Z^2} \mid x \vDash \phi\}$. The set $X_{\phi}$ is always shift-invariant but not always closed, and so not always a subshift (see Example \ref{exnonsub}). For a set of configurations $X$, if $X=X_{\phi}$ for some $\phi$, we say that $X$ is \emph{MSO-definable}, and \emph{$C$-definable} for a class of formulas $C$ if $\phi \in C$.

Every SFT is definable by an MSO formula of the form $\forall v \ \psi$, where $\psi$ is quantifier-free.
This is essentially the format used by Griddy \cite{diddyrepo}, with syntactic support for easily defining complex quantifier-free formulas.

If $X_{\phi}=X_{\psi}$, we say $\phi$ and $\psi$ are \emph{equivalent}. If $Q$ is a quantifier, then $\ob{Q}=\forall$ if $Q=\exists$, $\ob{Q}=\exists$ otherwise. We say that a quantifier $Q$ in a formula $\phi$ is \emph{irrelevant} if $\phi$ is equivalent to the formula where that quantifier has been replaced by $\ob{Q}$.

Our second order variables range over subsets of $\Z^2$.
It is equivalent to quantify over configurations in the following sense.
Consider a formula
$\exists X \in X_{\phi} \ \psi$,
where $X$ ranges over configurations of $X_{\phi}$ and $\psi$ may contain subformulas $P_c(X(t))$ with the interpretation that the term $t$ is coloured with the colour $c$ in the configuration $X$.
This formula be rewritten as an equivalent $\Sb{n}$ formula if $\phi, \psi \in \Sb{n}$, and analogously for $\Pb{n}$, by replacing $X$ with $|\A|$ subsets of $\Z^2$, one for each color, with the requirement that the sets form a partition of $\Z^2$. Thus, we allow quantification over configurations since it is often more convenient.
Furthermore, we may restrict any such quantification to be over a set of configurations defined by a first order formula $\psi$ without changing the second-order quantifier complexity: $\forall_1 X \in X_\psi \ \phi$ can be implemented as $\forall_1 X (\neg\psi(X) \vee \phi)$, and $\exists_1 X \in X_\psi \ \phi$ as $\exists_1 X (\psi(X) \wedge \phi)$, after which the quantifiers of $\phi$ can be moved over $\psi(X)$ to bring it to prenex normal form.

%We can get half a page removing these examples if we need, but I think they are useful.

\begin{example}
  \label{ex:finite}
We define $\mathrm{fin}(E)$, a $\Sb{1}$ formula with a free second order variable $E$ which is true if and only if the set $E \subset \Z^2$ is finite.
The formula is $\mathrm{fin}(E)\ed \exists N \exists S \ \phi_1 \land \phi_2 \land \phi_3 \land \phi_4 \land \phi_5$, where
\begin{itemize}
    \item $\phi_1 \ed \forall v \ (v \in N \Leftrightarrow \west(v) \in N \land \south(v) \in N)$

    \item $\phi_2 \ed \forall v \ (v \in S \Leftrightarrow \east(v) \in S \land \north(v) \in S)$

    \item $\phi_3 \ed \exists v \ v \in N \land \lnot \north(v) \in N \land \lnot \east(v) \in N$

    \item $\phi_4 \ed \exists v \ v \in S \land \lnot \south(v) \in S \land \lnot \west(v) \in S$

    \item $\phi_5 \ed \forall v \ v \in E \Rightarrow v \in N \land v \in S$
\end{itemize}
The idea is that $N \subset \Z^2$ is a nondegenerate southwest quarter-plane, $S \subset \Z^2$ is a nondegenerate northeast quarter-plane, and $E$ is contained in their intersection.
\end{example}

\begin{example}
We define $\mathrm{Stable}(X) = \forall n \ n \in X \Leftrightarrow \east(n) \in X$ where $X$ is a second order variable. The predicate $\mathrm{Stable}(X)$ is valuated to true if, whenever an element is in $X$, then all the elements of the same line are in $X$ too. We define also $\mathrm{Line}(X,Y)$ as the formula
\[
\mathrm{Stable}(X) \land ((\mathrm{Stable}(Y) \land (\exists n \ n \in X \land n \in Y)) \Rightarrow (\forall n \ n \in X \Rightarrow n \in Y)).
\]
Then $\forall Y \ \mathrm{Line}(X,Y)$ is true if and only if $X$ is exactly the set of elements of a line, or is empty. Finally, define the formula $\phi=\exists X \ \forall Y \ \mathrm{Line}(X,Y) \land (\forall n \ P_{\bsq}(n) \Rightarrow n \in X)$.
The subshift $X_{\phi}$ consists of exactly those configurations where every symbol $\bsq$ appears in the same line.
\end{example}

\begin{example}\label{exnonsub}
Define the formula $\phi=\exists x \ \forall y \ P_{\bsq}(y) \Leftrightarrow y=x$. Then $X_{\phi}$ is the set of configurations where there is exactly one position colored $\bsq$. It is not a subshift since it is not topologically closed. Indeed, there is a sequence of configurations in $X_{\phi}$ which converges to the all-$\square$ configuration, which is not in $X_{\phi}$.
\end{example}

\section{First order formulas}

In this section, we analyze the structure of FO-definable sets.
The main result is that they are exactly the sets that can be defined by restricting the number of occurrences of a finite number of finite patterns up to some finite upper bound.
Our techniques come from \cite{rect}, where similar results were proved in the context of \emph{picture languages}.

\begin{definition}
Let $k \geq 0$.
For $0 \le a \le k$ and a finite pattern $P$ on an alphabet $\A$, define $S^k_a(P) \subset \A^{\Z^2}$ as
the set of configurations where $P$ appears exactly $a$ times if $a<k$, or at least $k$ times if $a=k$.

Let also $n \geq 0$.
For a function $f : \A^{n \times n} \to [0, k]$, define $C_f = \bigcap_P S^k_{f(P)}(P)$.
Then $\{ C_f \mid f : \A^{n \times n} \to [0, k] \}$ is a partition of $\A^{\Z^2}$ into equivalence classes, and if two configurations $x, y \in \A^{\Z^2}$ are in the same class, we denote $x \sim_{(n,k)} y$.
\end{definition}

In other words, $x \sim_{(n,k)} y$ means that each $n \times n$ pattern either occurs exactly $t < k$ times in both $x$ and $y$, or occurs at least $k$ times in both.

\begin{lemma}\label{equivclass}
Let $\phi$ be a first order formula. There exists a couple $(n,k) \in \N^2$ such that $X_{\phi}$ is a union of $\sim_{(n,k)}$-equivalence classes. Moreover, this union is computable from $\phi$.
\end{lemma}

We stress that apart from the computability of the union, this result essentially appeared already in \cite{j-t,rect}.

\begin{corollary} \label{fofonda}
Let $\phi$ be a first order formula. We can compute two formulas of the form $\exists x_1 \ \hdots \ \exists x_n \ \forall y_1 \ \hdots \ \forall y_m \ \psi_1$ and $\forall x_1 \ \hdots \ \forall x_n \ \exists y_1 \ \hdots \ \exists y_m \ \psi_2$ with $\psi_1$ and $\psi_2$ quantifier-free which are equivalent with $\phi$.
\end{corollary}

\begin{proof}
For all $n, k \in \N$, $0 \ge a \ge k$ and $P$ pattern of size $n \times n$, $S_a(P)$ is first-order definable by formulas of the form $\exists x_1 \ \hdots \ \exists x_n \ \forall y_1 \ \hdots \ \forall y_m \ \psi_1$ and of the form $\forall x_1 \ \hdots \ \forall x_n \ \exists y_1 \ \hdots \ \exists y_m \ \psi_2$, with $\psi_1$ and $\psi_2$ quantifier-free, and these formulas are computable from $S_a(P)$. It follows that it is also the case for every first-order formula using Lemma \ref{equivclass}.
\end{proof}

\section{MSOSUB}
\label{sec:msosub}

In this section, we are interested in the following problems:

\begin{definition}[MSOSUB]
Let $C$ be a class of MSO formulas.
We define the following problem, called $C$-SUB: given a formula $\phi \in C$, is $X_{\phi}$ a subshift?
\end{definition}

If $C$ is $\Sb{0}$ or $\Pb{0}$, we prefer to call this problem $\MFOSUB$.

\subsection{FOSUB}

\begin{theorem}
$\MFOSUB$ is $\Pi^0_4$-complete.
\end{theorem}

We split the proof into Lemmas \ref{mfoin} and \ref{mfohard}.

% \begin{definition}
% Given a $\sim_{(n,k)}$ equivalence class $C=\iC{\A}$, we denote by $f_C: \A^{n^2} \rightarrow [k]$ the function $f_C(P)=a_P$.
% A \emph{map} of $C$ is a function $\mathrm{map}: \A^{n^2} \rightarrow \Pe_{\mathrm{fin}}(\Z^2)$ such that $|\mathrm{map}(P)| \le f_C(P)$ for all $P$.
% A \emph{$C$-partition} is a pair of functions $I_1, I_2: \A^{n^2} \rightarrow [k]$ such that $I_1(P)+I_2(P)=f_C(P)$ for all $P$.
% For $l \in \N$ and two maps $\mathrm{map}_1, \mathrm{map}_2$ of $C$, we say that $(C,l,\mathrm{map}_1,\mathrm{map}_2)$ \emph{tiles the plane} if there exists $x \in \A^{\Z^2}$ such that for each $P \in \A^{n^2}$:
% \begin{itemize}
%     \item If $S_a(P) \in C$ for $a < k$, then the pattern $P$ appears at positions $\mathrm{map}_1(P) \cup \mathrm{map}_2(P)$ and no other positions,
%     \item if $S_k(P) \in C$, then the pattern $P$ appears on positions $\mathrm{map}_1(P) \cup \mathrm{map}_2(P)$ (with $|\mathrm{map}_1(P) \cup \mathrm{map}_2(P)|\ge k$) and if $|\mathrm{map}_1(P)|<k$, the only positions on which a $P$ appears in the central square of size $l \times l$ are in $\mathrm{map}_1(P)$.
% \end{itemize}
% \end{definition}

% Informally, a map of an equivalence class $C$ records some of the positions of occurrence of each pattern $P \in \A^{n^2}$ in a configuration $x \in C$, and a partition of $C$ splits there occurrences into two classes.

\begin{lemma} \label{mfoin}
$\MFOSUB$ is $\Pi^0_4$.
\end{lemma}

%I removed the explanation of why F(UC_i) iff UC_i is not a subshift

\begin{proof}[sketch]
  Let $\phi$ be a first order formula.
  Using Lemma \ref{equivclass} we can compute $n, k \in \N$ and disjoint $\sim_{(n,k)}$-equivalence classes $C_1, \ldots, C_r$ such that $X_{\phi}=\bigcup_{i=1}^r C_i$. We want to prove that knowing whether $\bigcup_{i=1}^r C_i$ is closed is $\Pi^0_4$.

  For a $\sim_{(n,k)}$-equivalence class $C=\iC{\A}$, define the function $f_C : \A^{n^2} \to \{0, \ldots, k\}$ by $f_C(P)=a_P$.
  A \emph{map} of $C$ is a function $M: \A^{n^2} \rightarrow \Pe_{\mathrm{fin}}(\Z^2)$. % such that $|M(P)| \le f_C(P)$ for all $P$.
  For $\ell \in \N$ and two maps $M_1, M_2$ of $C$, we say that $(C,\ell,M_1,M_2)$ \emph{tiles the plane} if there exists $x \in \A^{\Z^2}$ such that for each $P \in \A^{n^2}$:
\begin{itemize}
    \item if $f_C(P) < k$, then the pattern $P$ appears at positions $M_1(P) \cup M_2(P)$ and no other positions,
    \item if $f_C(P) = k$, then the pattern $P$ appears on positions $M_1(P) \cup M_2(P)$ (with $|M_1(P) \cup M_2(P)|\ge k$), and if $|M_1(P)|<k$, the only positions of $[-\ell,\ell-1]^2$ at which $P$ appears are in $M_1(P)$.
    \end{itemize}
    To decide whether a given quadruple tiles the plane is $\Pi^0_1$.
  
  Now let us consider the following arithmetical formula $F$:
  there exist $m \in \N$, an index $i \in \{1, \ldots, r\}$ and two functions $I_1, I_2 : \A^{n^2} \to \{0, \ldots, k\}$ such that $I_1 + I_2 = f_{C_i}$, and
  for all $\ell > m$ there exist two maps $M_1, M_2$ of $C_i$ such that
  \begin{itemize}
  \item $|M_j(P)| = I_j(P)$ for each $j = 1,2$ and $P \in \A^{n^2}$,
  \item $M_1(P) \subseteq [-m,m]^2$ and $M_2(P) \cap [-\ell,\ell-1]^2 =\emptyset$ for each $P \in \A^{n^2}$,
  \item $(C_i,\ell,M_1,M_2)$ tiles the plane, and
  \item there is no $j \in \{1, \ldots, r\}$ such that $(C_j,0,M_1,P \mapsto \emptyset)$ tiles the plane.
  \end{itemize}

  By its form, $F$ is $\Sigma^0_4$. One can prove that $F$ holds if and only if $\bigcup_{i=1}^r C_i$ is not closed.
  The idea is that if $F$ holds, then one obtains a sequence $x_\ell$ of configurations in which the maps $M_1$ and $M_2$ control the positions of each pattern $P \in \A^{n^2}$, and the tiling conditions make sure that at least one pattern ``escapes to infinity'' in such a way that no limit point of the sequence is in $\bigcup_{i=1}^r C_i$.
\end{proof}

\begin{lemma} \label{mfohard}
$\MFOSUB$ is $\Pi^0_4$-hard.
\end{lemma}

\begin{proof}[sketch]
We reduce the $\Pi^0_4$-complete problem $\forall COF$: given a Turing Machine $M$ which takes two inputs, does $M(v,\_)$ have a cofinite language for all $v$?
For this, we define an SFT $X$ %(which is definable by a first order formula)
and add on it a first order formula to get a formula $\phi$. Then, we will prove that $M \in \forall COF$ if and only if $\phi \in \MFOSUB$.

Let $M$ be a turing machine. The SFT $X$ is the one described by Figure \ref{fig:mfohard} (the only allowed patterns are those of size $2 \times 2$ in this figure, but other inputs $v$ and $w$ are possible, even the infinite one). %We use standard techniques to simulate $M$ (see Property \ref{computationtiles}).
The input $v$ is copied to the southeast until it hits the gray diagonal signal emitted by the ${*}$-symbol at the end of the second input $w$, where it is erased.
The diagonal turns black after the point of erasure.
%Configurations of $X$ with two finite inputs $x$ and $y$ and a simulated computation are called \emph{classical configurations}.
%but keep in mind that there are many others since it is a subshift.
The machine $M$ is not allowed to halt, so in all configurations containing two finite inputs $v$ and $w$, $M(v,w)$ does not halt. 

Let $\phi'$ be the formula which defines this SFT, and let $\phi=\phi' \land ((\exists x \ \exists y \ P_1(x) \land P_2(y)) \Rightarrow \exists z \ P_3(z))$, where tiles 1, 2 and 3 are showed in Figure \ref{fig:mfohard}). Intuitively, $\phi$ requires the configuration to satisfy $\phi'$, and if there is a diagonal with a finite input $v$ written on it, then there also exists a finite input $w$.

Now $M \in \forall COF$ if and only if $\phi \in MFOSUB$, since the only reason for $X_\phi$ not to be closed is that for some input $v$, there are infinitely many inputs $w$ such that $M(v,w)$ does not halt, and a limit point of the correspondng configurations violates the second condition of $\phi$.
\end{proof}

\begin{figure}
    \centering
    \begin{tikzpicture}[scale=0.5]

\begin{scope}
\draw[clip] decorate [decoration={random steps,segment length=4pt,
amplitude=1.5pt}] {(-0.5,-0.5) -- (16.5,-0.5) -- (16.5,12.5) -- (-0.5,12.5) -- cycle};

\fill[white!70!black] (0,8) -- (17,8) -- (17,15) -- (0,15) -- cycle;
\fill[\grey] (13,8) -- (14,8) -- (6,0) -- (5,0) -- cycle;
\fill (7,2) -- (7.5,1.5) -- (5,-1) -- (4,-1) -- cycle;
%\fill[white!70!black] (1,8) -- (5,8) -- (5,9) -- (1,9) -- cycle;
%\fill[\grey] (0,8) -- (1,8) -- (1,9) -- (0,9) -- cycle;
%\fill[\grey] (5,8) -- (6,8) -- (6,9) -- (5,9) -- cycle;
%\fill[\grey] (13,8) -- (14,8) -- (14,9) -- (13,9) -- cycle;
\fill[\grey] (6,8) -- (13,8) -- (13,9) -- (6,9) -- cycle;

\draw (0,-1) -- (0,15);
\draw (1,-1) -- (1,15);
\draw (2,-1) -- (2,15);
\draw (3,-1) -- (3,15);
\draw (4,-1) -- (4,15);
\draw (5,-1) -- (5,15);
\draw (6,-1) -- (6,15);
\draw (7,-1) -- (7,15);
\draw (8,-1) -- (8,15);
\draw (9,-1) -- (9,15);
\draw (10,-1) -- (10,15);
\draw (11,-1) -- (11,15);
\draw (12,-1) -- (12,15);
\draw (13,-1) -- (13,15);
\draw (14,-1) -- (14,15);
\draw (15,-1) -- (15,15);
\draw (16,-1) -- (16,15);

\draw(-1,0) -- (17,0);
\draw(-1,1) -- (17,1);
\draw(-1,2) -- (17,2);
\draw(-1,3) -- (17,3);
\draw(-1,4) -- (17,4);
\draw(-1,5) -- (17,5);
\draw(-1,6) -- (17,6);
\draw(-1,7) -- (17,7);
\draw(-1,8) -- (17,8);
\draw(-1,9) -- (17,9);
\draw(-1,10) -- (17,10);
\draw(-1,11) -- (17,11);
\draw(-1,12) -- (17,12);
\draw(-1,13) -- (17,13);
\draw(-1,14) -- (17,14);

\draw(0.5,8.5) node {$\$$};
\draw(1.5,8.5) node {$x_1$};
\draw(2.5,8.5) node {$x_2$};
\draw(3.5,8.5) node {$x_3$};
\draw(4.5,8.5) node {$x_4$};
\draw(5.5,8.5) node {$\#$};
\draw(6.5,8.5) node {$y_1$};
\draw(7.5,8.5) node {$y_2$};
\draw(8.5,8.5) node {$y_3$};
\draw(9.5,8.5) node {$y_4$};
\draw(10.5,8.5) node {$y_5$};
\draw(11.5,8.5) node {$y_6$};
\draw(12.5,8.5) node {$y_7$};
\draw (13.5,8.5) node {$*$};

\draw(2.5,7.5) node {$x_1$};
\draw(3.5,6.5) node {$x_1$};
\draw(4.5,5.5) node {$x_1$};
\draw(5.5,4.5) node {$x_1$};
\draw(6.5,3.5) node {$x_1$};
\draw(7.5,2.5) node {$x_1$};

\draw(3.5,7.5) node {$x_2$};
\draw(4.5,6.5) node {$x_2$};
\draw(5.5,5.5) node {$x_2$};
\draw(6.5,4.5) node {$x_2$};
\draw(7.5,3.5) node {$x_2$};
\draw(8.5,2.5) node {$x_2$};

\draw(4.5,7.5) node {$x_3$};
\draw(5.5,6.5) node {$x_3$};
\draw(6.5,5.5) node {$x_3$};
\draw(7.5,4.5) node {$x_3$};
\draw(8.5,3.5) node {$x_3$};

\draw(5.5,7.5) node {$x_4$};
\draw(6.5,6.5) node {$x_4$};
\draw(7.5,5.5) node {$x_4$};
\draw(8.5,4.5) node {$x_4$};
\draw(9.5,3.5) node {$x_4$};

\fill[white!70!black] (2.5,10.5) -- (11.5,10.5) -- (11.5,11.5) -- (2.5,11.5) -- cycle;
\draw (7,11) node {\Large Computation of $M$};

\end{scope}

\draw (8.5,1.5) node {2};
\draw[->] (8.2,1.5) -- (7.8,1.5);

\draw (14.5,7.5) node {1};
\draw[->] (14.2,7.5) -- (13.8,7.5);

\draw (-1,8.5) node {3};
\draw[->] (-0.7,8.5) -- (0.2,8.5);

\end{tikzpicture}
    \caption{A configuration of the SFT $X$ in the proof of Lemma \ref{mfohard}.}
    \label{fig:mfohard}
\end{figure}
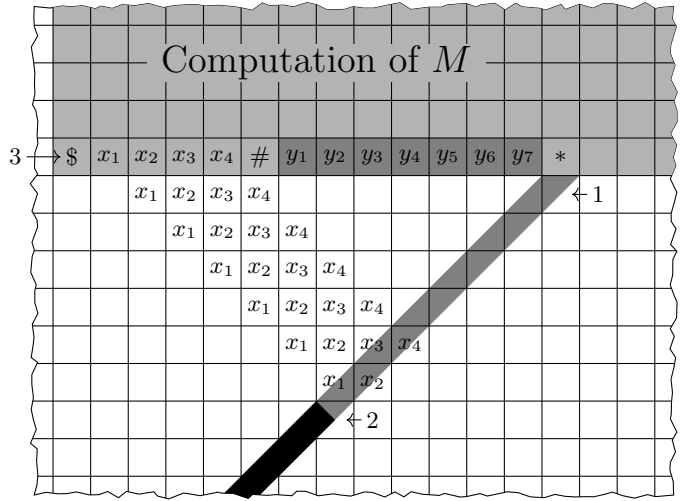

\subsection{Higher order cases}

\begin{theorem}\label{PbnSUB}
For all $n \ge 1$, $\Pb{n}$-SUB and $\Sb{n}$-SUB are $\Pi^1_n$-complete.
\end{theorem}

We split the proof into Lemmas \ref{pisubin}, \ref{pisubhard} and \ref{sisubhard}.

\begin{remark}\label{calcfonda}
All problems of the form $\exists_1 X_1 \exists x_1 \forall x_2 \ R$ with $R$ a decidable predicate, are in $\Sigma^0_2$. Dually, all problems of the form $\forall_1 X_1 \forall x_1 \exists x_2 \ R$ with $R$ a decidable predicate, are in $\Pi^0_2$.
\end{remark}

%I removed the proof of the previous remark.

\begin{lemma} \label{pisubin}
For all $n \ge 1$, $\Pb{n}$-SUB and $\Sb{n}$-SUB are in $\Pi^1_n$.
\end{lemma}

%In the following proof, I removed the part where we use the main argument to conclude the lemma.

\begin{proof}%[sketch]
  Let $\phi \in \Pb{n}$ with $m$ second order quantifiers (the $\Sb{n}$ case is similar). Denote by $\phi^*$ the first order formula on $\A \times \{0,1\}^m$ obtained from $\phi$ by removing second order quantifiers and replacing each $v \in X_i$ with $\bigvee_{c \in \pi_i^{-1}(1)} P_c(v)$, where $\pi_i (c_0, \hdots, c_m)=c_i$. For a configuration $x \in \A^{\Z^2}$, we have $x \vDash \phi$ if and only if
  \begin{equation}
    \label{eq:pisubin}
    \forall X_1 \ \hdots \ Q X_m \ x \times X_1 \times \hdots \times X_m \vDash \phi^*
  \end{equation}
for $Q=\exists$ if $n$ is even, $Q=\forall$ otherwise.
From Corollary \ref{fofonda} we deduce that $x \times X_1 \times \hdots \times X_m \vDash \phi^*$ is a $\Delta^0_2$ condition, when $\phi^*$ is given as input and $x$ and the $X_i$ are given as oracles.
With Remark \ref{calcfonda} it follows that \eqref{eq:pisubin}, and hence $x \vDash \phi$, is $\Pi^0_2$ for $n = 1$, and $\Pi^1_{n-1}$ for $n \geq 2$, with $\phi$ as input and $x$ as an oracle.

The set $X_\phi$ is closed if and only if for all sequences $x_n$ of configurations such that $x_n \vDash \phi$ for all $n$ and $x_n |_{[-n,n]^2} = x_m|_{[-n,n]^2}$ for all $m \geq n$, we have $\lim_n x_n \vDash \phi$.
By the above, this condition is $\Pi^1_n$.
\end{proof}

\begin{lemma} \label{pisubhard}
For all $n \ge 1$, $\Pb{n}$-SUB is $\Pi^1_n$-hard.
\end{lemma}

\begin{proof}
  We reduce an arbitrary problem in $\Pi^1_n$. Let $R$ be a computable predicate. We construct a formula $\phi \in \Pb{n}$ such that $\forall A_1 \hdots Q A_n \ \ob{Q} k \  Q m \ R$ holds if and only if $X_\phi$ is a subshift, with $Q=\exists$ if $n$ is even and $Q=\forall$ otherwise.

  Let $\phi'$ be the formula which describes a tiling where a northeast quarter-plane is coloured with 0 and 1 such that each row is coloured in a same way, and the other tiles are blank. Since the quarter-plane must appear, $X_{\phi'}$ is not a subshift. The infinite binary sequence repeated on each line is interpreted as the oracle $A_1$ in the following.

  Let $\phi=\forall X_2 \in X_{\phi'} \hdots \ \ob{Q} X_n \in X_{\phi'} \ \ob{Q} Z \ \phi' \wedge \psi$, where $\psi$ is the formula stating the following:
  \begin{itemize}
  \item each of $X_2, \ldots, X_n$ has a northeast quarter-plane at the same position as the main configuration $x \in X_{\phi'}$,
  \item the configuration $Z$ contains, in the same quarter-plane, a simulated computation of a Turing machine $M$ that can use the binary sequences of $x$ and the $X_i$ as oracles,
  \item if $\ob{Q} = \exists$, then $M$ visits a special state $q_s$ an infinite number of times (a $\Pb{1}$ condition by Example \ref{ex:finite}), and
  \item if $\ob{Q} = \forall$, then $M$ visits the state $q_s$ a finite number of times.
  \end{itemize}
  Notice that the configuration $Z$ is determined by what was quantified before, so its quantifier is irrelevant.
  Hence $\phi$ is $\Pb{n}$.

We specify the machine $M$.
Starting with $k = 0$, it enumerates $m = 0, 1, 2, \ldots$ until it finds one such that $R(A_1, \ldots, A_n, k, m)$ does not hold. Whan this happens, it visits the state $q_s$, increments $k$, and starts the enumeration again.

Notice that $x \vDash \phi$ if and only if $\forall A_2 \hdots \ \ob{Q} A_n \ \ Q k \ \ob{Q} m \ \lnot R$ holds.
Hence, if $\forall A_1 \exists A_2 \hdots Q A_n \ \ob{Q}k \ Q m \ R$ holds, then $X_{\phi}$ is empty, and so a subshift.
If the formula does not hold, then at least one configuration with a quarter-plane is in $X_{\phi}$. But the blank configuration is not in $X_\phi$, so $X_{\phi}$ is not a subshift.
\end{proof}

\begin{lemma} \label{sisubhard}
For all $n \ge 1$, $\Sb{n}$-SUB is $\Pi^1_n$-hard.
\end{lemma}

The proof is similar to that of Lemma \ref{pisubhard}, but instead of emptiness, we use a separate binary layer with at most one $1$-symbol.

\section{Complexity of the MSO languages}
\label{sec:complexity}

In this section, we consider the complexity of the $C$-definable sets and subshifts for $C$ classes of the MSO hierarchy.

\subsection{Maximal complexity}

Here we determine the maximal complexity of languages of formulas in the different classes of the MSO hierarchy.

\begin{theorem}
Languages of FO and $\Sb{1}$ sets are $\Sigma^0_2$.
Languages of $\Sb{n}$ sets for $n \ge 2$ are $\Sigma^1_{n-1}$.
Languages of $\Pb{n}$ sets for $n \ge 1$ are $\Sigma^1_n$.
\end{theorem}

\begin{proof}
Let $\phi \in \Sb{n}$. For a pattern $P$, we have $P \in \Le(X_{\phi})$ if there exists $x \in \A^{\Z^2}$ such that $x \vDash \phi$ and $P$ appears in $x$. Given $x$ as an oracle, $P$ appearing in $x$ is a $\Sigma^0_1$ predicate, and $x \vDash \phi$ is $\Sigma^1_{n-1}$ if $n \ge 2$, and $\Sigma^0_2$ if $n \le 1$ thanks to Corollary \ref{fofonda} and Remark \ref{calcfonda}. We conclude that $\Le(X_{\phi})$ is $\Sigma^1_{n-1}$ if $n \ge 2$ and $\Sigma^0_2$ if $n \le 1$ (we need once more Remark \ref{calcfonda} for the case $n=1$).

Let $\phi \in \Pb{n}$ for $n\ge 1$. Given $x \in \A^{\Z^2}$ as an oracle, $x \vDash \phi$ is $\Pi^1_{n-1}$ if $n \ge 2$, and $\Pi^0_2$ if $n=1$ thanks to Corollary \ref{fofonda} and Remark \ref{calcfonda}. We conclude that $\Le(X_{\phi})$ is $\Sigma^1_n$ for all $n \ge 1$.
\end{proof}

\begin{theorem} \label{folim}
There exists an FO-definable subshift with a $\Sigma^0_2$-hard language.
\end{theorem}

\begin{proof}
Let $\phi'$ be the first order formula which defines an SFT $X$ containing, in a northeast quarter-plane, a simulated computation of a turing machine $M'$. Figure \ref{fig:M'} contains asample configuration. The machine $M'$ takes the code $[M]$ of another Turing machine $M$, and starting with $k = 0$, simulates $M$ with input $k$. If $M$ halts on $k$, then $M'$ enters a special state $q_H$, increments $k$ and starts again. The tile containing $q_H$ sends a signal (dark gray in Figure \ref{fig:M'}) to the west, which turns south when it encounters the left side of the quarter-plane, merging with another southward signal if one is present. We forbid to have such a signal under the diagonal of the quarter-plane.

Now let $\phi=\phi' \land ((\exists u \ \exists v \ P_1(u) \land P_2(v) \land P_3(\south(v))) \Rightarrow \exists w \ \lnot P_4(w) \land P_3(\west(w)))$ (tiles 1, 2, 3 and 4 are shown on Figure \ref{fig:M'}). Intuitively, $\phi$ requires that $\phi'$ is true and if the computation starts and the input is finite, then there exists a tile on the west border of the quarter-plane with no signal on it.

We prove first that $X_{\phi}$ is a subshift. Suppose for a contradiction that there is a sequence $(x_n)_{n \in \N}$ such that $x_n \vDash \phi$ for all $n$ and $x_n \overset{n \rightarrow \infty}{\longrightarrow} x \not \vDash \phi$. Since $X_{\phi'}$ is a subshift, $x \vDash \phi'$. Hence $x$ shows a computation with a Turing machine $M$ on which the west border of the quarter-plane is covered by the signal. As $x_n$ converges to $x$, from some $n$ onward, the code $[M]$ is entirely written on the south border of the quarter-planes of the $x_n$ and the computation of $M'$ is determined by this. Since $x_n \vDash \phi$, the west borders are not covered by the signal. Then the the sequence $x_n$ must be eventually constant. With $x \not \vDash \phi$ we reach a contradiction, so $X_{\phi}$ is a subshift.

Notice that $[M]$ is in $\Le(X_{\phi})$ if and only if $M \in \coTOTAL$, i.e. if $M$ does not halt on at least one input. This problem is known to be $\Sigma^0_2$-complete. 
\end{proof}

\begin{figure}
    \centering
\begin{tikzpicture}[scale=0.5]

\begin{scope}
\draw[clip] decorate [decoration={random steps,segment length=4pt,
amplitude=1.5pt}] {(-0.5,-0.5) -- (-0.5,10.5) -- (10.5,10.5) -- (10.5,-0.5) -- cycle};

\fill[white!80!black] (0,0) -- (0,11) -- (11,11) -- (11,0) -- cycle;

\fill[\grey] (0,0) -- (1,0) -- (1,11) -- (0,11) -- cycle;
\fill[\grey] (1,4) -- (4,4) -- (4,5) -- (1,5) -- cycle;
\fill[\grey] (1,6) -- (3,6) -- (3,7) -- (1,7) -- cycle;
\fill[\grey] (1,8) -- (5,8) -- (5,9) -- (1,9) -- cycle;

\draw (-1,0) -- (11,0);
\draw (-1,1) -- (11,1);
\draw (-1,2) -- (11,2);
\draw (-1,3) -- (11,3);
\draw (-1,4) -- (11,4);
\draw (-1,5) -- (11,5);
\draw (-1,6) -- (11,6);
\draw (-1,7) -- (11,7);
\draw (-1,8) -- (11,8);
\draw (-1,9) -- (11,9);
\draw (-1,10) -- (11,10);

\draw (0,-1) -- (0,11);
\draw (1,-1) -- (1,11);
\draw (2,-1) -- (2,11);
\draw (3,-1) -- (3,11);
\draw (4,-1) -- (4,11);
\draw (5,-1) -- (5,11);
\draw (6,-1) -- (6,11);
\draw (7,-1) -- (7,11);
\draw (8,-1) -- (8,11);
\draw (9,-1) -- (9,11);
\draw (10,-1) -- (10,11);

%\fill[white!80!black] (6.1,6.1) -- (5.5,7) -- (9.7,11) -- (9.5,9.8) -- cycle;

\draw[very thick] (0,0) -- (10.6,10.6); %node[near end, above, sloped, fill=white!80!black] {Computation of $M'$};

%\draw[red] (3.5,4.5) node {$q_H$};

%\draw[red] (2.5,6.5) node {$q_H$};

%\draw[red] (4.5,8.5) node {$q_H$};

\draw (0.5,0.5) node {$\$$};

\draw (1.5,0.5) node {$M_1$};
\draw (1.5,1.5) node {$M_1$};
\draw (2.5,0.5) node {$M_2$};
\draw (2.5,1.5) node {$M_2$};
\draw (2.5,2.5) node {$M_2$};
\draw (3.5,0.5) node {$M_3$};
\draw (3.5,1.5) node {$M_3$};
\draw (3.5,2.5) node {$M_3$};
\draw (3.5,3.5) node {$M_3$};
\draw (4.5,0.5) node {$M_4$};
\draw (4.5,1.5) node {$M_4$};
\draw (4.5,2.5) node {$M_4$};
\draw (4.5,3.5) node {$M_4$};
\draw (4.5,4.5) node {$M_4$};
\draw (5.5,0.5) node {$M_5$};
\draw (5.5,1.5) node {$M_5$};
\draw (5.5,2.5) node {$M_5$};
\draw (5.5,3.5) node {$M_5$};
\draw (5.5,4.5) node {$M_5$};
\draw (5.5,5.5) node {$M_5$};

\draw (0.5,0.5) -- (0.5,1.5) -- (1.5,1.5) -- (1.5,2.5) -- (2.5,2.5) -- (2.5,3.5) -- (3.5,3.5) -- (3.5,4.5) -- (2.5,4.5) -- (2.5,6.5) -- (3.5,6.5) -- (3.5,7.5) -- (4.5,7.5) -- (4.5,8.5) -- (3.5,8.5) -- (3.5,9.5) -- (2.5, 9.5) -- (2.5,11);

\draw[dashed] (0.5,1) -- (0.5, 11);
\draw[dashed] (1.5,2) -- (1.5, 11);
\draw[dashed] (2.5,3) -- (2.5, 11);
\draw[dashed] (3.5,4) -- (3.5, 11);
\draw[dashed] (4.5,5) -- (4.5, 11);
\draw[dashed] (5.5,6) -- (5.5, 11);

\stateb{0.5}{1.5}{1}
\stateb{1.5}{2.5}{2}
\stateb{2.5}{3.5}{3}
\stateb{3.5}{4.5}{H}%{red}
\stateb{2.5}{5.5}{4}
\stateb{2.5}{6.5}{H}%{red}
\stateb{3.5}{7.5}{5}
\stateb{4.5}{8.5}{H}%{red}
\stateb{3.5}{9.5}{6}
\stateb{2.5}{10.5}{7}

\end{scope}

\draw (-1,0.5) node {1};
\draw[->] (-0.7,0.5) -- (0.2,0.5);

\draw (11,0.5) node {2};
\draw[->] (10.7,0.5) -- (7.6,0.5);

\draw (-1,2) node {4};
\draw[->] (-0.8,2) -- (0.2,2.5);
\draw[->] (-0.8,2) -- (0.2,1.5);

\draw (-1,4.5) node {3};
\draw[->] (-0.8,4.5) -- (-0.35,4.5);

\end{tikzpicture}
    \caption{Computation of $M'$.}
    \label{fig:M'}
\end{figure}
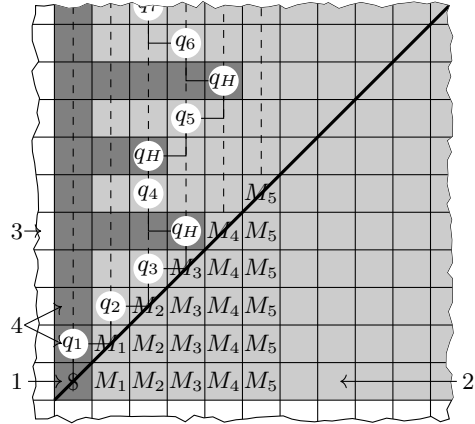

\begin{theorem}\label{Pbnlim}
For all $n \ge 1$, there exists a $\Pb{n}$-definable set with a $\Sigma^1_n$-hard language.
\end{theorem}

\begin{proof}
Let $\phi'$ be the first order formula which defines the set of configurations that are blank except for an infinite binary word $w_1$, and the code $[M]$ of a total Turing machine $M$ written just above this word, at the left end. As no other configurations are allowed, $X_{\phi'}$ is not a subshift.

Let $\phi = \phi' \wedge \forall X_2 \ \exists X_3 \hdots \ Q \ X_n \ Q \ Z \ \psi$, where $Q = \forall$ if $n$ is even and $Q = \exists$ otherwise, and $\psi$ specifies that each $X_i$ contains an infinite binary word $w_i$ at the same position as the main configuration $x$, and $Z$ contains a simulated computation which verifies $\exists k \ \forall m \ M(w_1, \ldots, w_n, k, m)$ if $n$ is even, and $\forall k \ \exists m \ M(w_1, \ldots, w_n, k, m)$ if $n$ is odd.
We use the same technique as in Lemma \ref{pisubhard} to design $\phi$ so that it is $\Pb{n}$.

%Let $\phi$ be the $\Pb{n}$ formula which requires that $\phi'$ is satisfied and $\forall X_2 \ \exists X_3 \hdots \ \forall X_n \ \forall$ configurations (resp. $\exists X_n \ \exists$ a configuration such that) then exists $x_1$ such that $M(x,x_1,x_2,X_1,\hdots,X_n)$ does not halt (resp. for all $x_1$ $M(x,x_1,x_2,X_1,\hdots,X_n)$ halts) if $n$ is even (resp. if $n$ is odd). $X_1$ is interpreted by the 0 and 1 line. We use the same technique as in Lemma \ref{pisubhard} to define $\phi$. We can verify that $\phi$ is in $\Pb{n}$.

Now $[M]$ is in $\Le(X_{\phi})$ if and only if
\[
\exists w_1 \ \forall w_2 \ \hdots \ Q \ w_n \ \ob{Q} \ k \ Q \ m \ M(w_1, \ldots, w_n, k, m).
\]
This condition is $\Sigma^1_n$-hard, so $\Le(X_{\phi})$ is $\Sigma^1_n$-hard too.
\end{proof}

\begin{theorem}\label{pb1language}
There exists a $\Pb{1}$-definable subshift with a $\Pi^0_3$-hard language.
\end{theorem}

%In the following proof, I removed the part where we prove why it is a subshift and why it has a Pi^0_3-hard language.

\begin{proof}[sketch]
Let $X$ be the subshift in which all the configurations are blank except at most one infinite word, which contains the code $[M]$ of a Turing machine $M$ followed by an infinite sequence of tiles coloured with $\{0,1,|\}$ (plus all limit configurations of these). Let $\psi$ be the first order formula such that $X=X_{\psi}$.

Let $\phi_1$ be a $\Pb{1}$ formula which requires that if $[M]$ is finite, then there are infinitely many $|$-tiles. We interpret this as $[M]$ being followed by an infinite list of natural numbers in binary.

Let $\phi_2$ be a $\Pb{1}$ formula which requires that if $[M]$ is finite, then for each natural number $n$ in the list, $M$ never halts on $n$. This can be done by universally quantifying on a configuration containing a simulated computation of $M(n)$.

Let $\phi_3$ be a $\Pb{1}$ formula which requires that if $[M]$ is finite, then for each natural number $n$, either $n$ occurs in the list, or $M(n) \downarrow$. This can be done by simulating $M(n)$ and simultaneously searching for $n$ in the list, and requiring that this computation halts at some point.

Let $\phi_4$ be a $\Pb{1}$ formula which requires that if $[M]$ is finite, the list of natural numbers written after $M$ is strictly increasing. Once again, we can do this using a single universal quantifier.

Then let $\phi$ be a $\Pb{1}$ formula equivalent with $\psi \land \phi_1 \land \phi_2 \land \phi_3 \land \phi_4$ (just rewriting this formula in the prenex form). One can prove that $X_{\phi}$ is a subshift with a $\Pi^0_3$-hard language by reducing the language of Turing machines that diverge on infinitely many inputs, which is known to be $\Pi^0_3$-complete.
\end{proof}

Thanks to Theorems \ref{folim} and \ref{Pbnlim}, we conclude that there exists a $\Sb{1}$-definable subshift with a $\Sigma^0_2$-hard language and for all $n \ge 2$, there exists a $\Sb{n}$-definable subshift with a $\Sigma^1_{n-1}$-hard language.

\subsection{Definable language classes}

Here we are interested in knowing which language classes of the arithmetical and analytical hierarchy are included in the languages of sets or subshifts definable by MSO formulas according to their complexity.

\begin{theorem}
Some decidable languages are not $\Sb{1}$-definable by sets.
Some decidable languages are not $\Pb{1}$-definable by sets.
\end{theorem}

\begin{proof}
We prove here the $\Sb{1}$ case; the $\Pb{1}$ case is similar. The proof is based on the technique of \cite{kass-madden-13}.

Let $X$ be the subshift on $\A=\{0,1,\#\}$ which is the closure of the set of configurations with two square $\{0,1\}$-patterns of the same size aligned and at distance one where everything around these patterns is colored by $\#$ and where each pattern is the mirror image of the other one. We denote the mirror image of a pattern $P$ by $P^R$.
It is clear that $\Le(X)$ is decidable.

Suppose for a contradiction that there exists a $\Sb{1}$ formula $\phi = \exists Y \in \Sigma^{\Z^2} \ \psi \in \Sb{1}$, where $\psi$ is FO, such that $\Le(X)=\Le(X_{\phi})$. For $N \geq 1$ and $P \in \{0,1\}^{N^2}$, let $x_P \in X$ be the configuration where $x|_{[-N,-1] \times [0,N-1]} = P$ and $x|_{[0,N-1]^2} = P^R$. The pattern $\rho_P = x|_{[-N-1, N] \times [-1,N]}$, which contains $P$ and $P^R$ surrounded by $\#$-symbols, is in $\Le(X) = \Le(X_{\phi})$, so it appears in some configuration of $X_\phi$. We see that the only configuration of $X_{\phi}$ in which $\rho_P$ appears in the same position as in $x_P$ must be $x_P$ itself.

As $\psi$ is FO, by Lemma \ref{equivclass} there exist $n,k \in \N$ such that $X_{\psi}$ is a union of $\sim_{(n,k)}$-equivalence classes. For $P \in \{0,1\}^{N^2}$, let $y_P \in \Sigma^{\Z^2}$ be such that $x_P \times y_P \vDash \psi$. For $N$ large enough there exist $P\neq Q \in \{0,1\}^{N^2}$ such that the pattern $(x_P \times y_P)|_P$ has the same border of size $n$ as $(x_Q \times y_Q)|_Q$ and the same number of each $n \times n$ pattern counted up to $k$ (since $(2|\Sigma|)^{4Nn-4n^2}(k+1)^{(2|\Sigma|)^{n^2}} < 2^{N^2}$ for large $N$).

Now consider the configuration $z = (z_1, z_2)$ defined by $z_{\vec v} = (x_P \times y_P)_{\vec v}$ for $\vec v \notin [0, N-1]^2$ and $z_{\vec v} = (x_Q \times y_Q)_{\vec v}$ for $\vec{v} \in [0, N-1]^2$.
In other words, we have replaced the $P^R$-patch of $x_P$ with $Q^R$, and made the corresponding replacement on the $y$-layer as well.
Let us prove that $z \sim_{(n,k)} x_P \times y_P$. They are equal outside $C_\mathrm{in} = [n, N-n-1]^2$ (see Figure \ref{fig:nondef}) and they have the same number of patterns of size $n \times n$ counted up to $k$ inside $C_\mathrm{out} = \Z^2 \setminus [0, N-1]^2$. So if $z$ and $x_P \times y_P$ are not equivalent, this is due to a pattern of size $n \times n$ which overlaps $C_\mathrm{in}$ and $C_\mathrm{out}$, which is impossible.

Thus $z \in X_{\psi}$ and $z_1 \in X_{\phi} \subseteq X$. But $z_1$ consists of the patterns $P$ and $Q^R$ surrounded by $\#$-symbols, so it is not in $X$, and we have a contradiction.
\end{proof}

\begin{figure}
	\centering
\begin{tikzpicture}[scale=0.8]
\fill[white!80!black] (0,0) -- (3,0) -- (3,3) -- (0,3) -- cycle;
\fill[white!80!black] (4,0) -- (7,0) -- (7,3) -- (4,3) -- cycle;
\draw (0,0) -- (3,0) -- (3,3) -- (0,3) -- cycle;
\draw (4,0) -- (7,0) -- (7,3) -- (4,3) -- cycle;
\draw (0.5,0.5) -- (2.5,0.5) -- (2.5,2.5) -- (0.5,2.5) -- cycle;
\draw (4.5,0.5) -- (6.5,0.5) -- (6.5,2.5) -- (4.5,2.5) -- cycle;
\draw (1.5,1.5) node {\Huge $P$};
\draw (5.5,1.5) node {\Huge $P^R$};

\draw[<->] (4,4) -- (7,4);
\draw[<->] (8,0) -- (8,3);
\draw (5.5,4) node[above] {$N$};
\draw (8,1.5) node[right] {$N$};
\draw[<->] (5.5,2.5) -- (5.5,3);
\draw (5.5,2.75) node[right] {$n$};
\draw[<->] (6.5,1.5) -- (7,1.5);
\draw (6.75,1.5) node[above] {$n$};

\draw (5.5,0.5) node[above] {$C_\mathrm{in}$};
\draw (5.5,0) node[below] {$C_\mathrm{out}$};

\draw (-0.5,-0.5) node {$\#$};
\draw (3.5,-0.5) node {$\#$};
\draw (3.5,3.5) node {$\#$};
\draw (-0.5,3.5) node {$\#$};
\draw (7.5,-0.5) node {$\#$};
\draw (7.5,3.5) node {$\#$};

\draw[dashed] (-0.25,-0.5) -- (3.25,-0.5);
\draw[dashed] (-0.25,3.5) -- (3.25,3.5);
\draw[dashed] (3.75,-0.5) -- (7.25,-0.5);
\draw[dashed] (3.75,3.5) -- (7.25,3.5);
\draw[dashed] (-0.5,-0.25) -- (-0.5,3.25);
\draw[dashed] (3.5,-0.25) -- (3.5,3.25);
\draw[dashed] (7.5,-0.25) -- (7.5,3.25);

\draw (-2,1.5) node {$\rho_P:$};
\end{tikzpicture}
    \caption{Non $\Sb{1}$-definability of a decidable language}
    \label{fig:nondef}
  \end{figure}
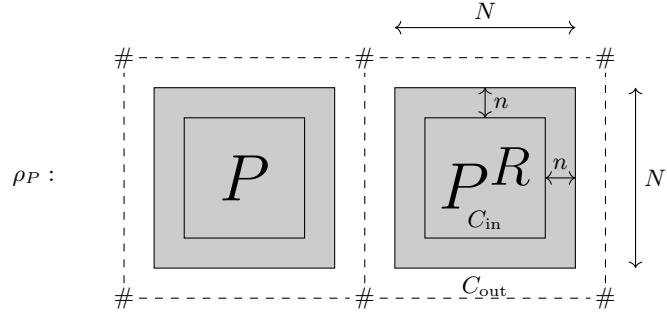

\begin{theorem}\label{Sbncompleteness}
For all $n \ge 2$, all $\Pi^1_{n-1}$ subshifts are $\Pb{n}$-definable, and all $\Sigma^1_{n-1}$ subshifts are $\Sb{n}$-definable.
\end{theorem}

\begin{proof}[sketch]
  We sketch the proof for $n = 3$.
  In the first case, we have a subshift $X \subseteq \A^{\Z^2}$ whose language is defined by a $\Pi^1_2$ formula $\forall_1 A_1 \exists_1 A_2 \forall k \exists m \xi$, where $\xi$ is computable.
  We construct an MSO formula $\phi = \forall X_1 \exists X_2 \forall X_3 \psi$, where the $X_i$ are constrained by FO conditions, that defines $X$.

  We follow the technique of \cite[Theorem 3]{torma}.
  The configuration $X_1$ contains a finite rectangle $R$, an encoding $w$ of a finite rectangular pattern over $\A$, and an infinite binary word $w_1$.
  The configuration $X_2$ contains either a gadget witnessing that $w$ is not the encoding of the pattern $x|_R$ in the main configuration $x \in \A^{\Z^2}$ (either due to having the wrong dimensions, or differing at a single coordinate), or an infinite binary word $w_2$.
  The configuration $X_3$ contains an infinite simulated computation of a Turing machine that can access the infinite words $w_1$ and $w_2$ as oracle tapes.
  The machine checks sequentially for $k = 0, 1, 2, \ldots$ that there exists $m \in \N$ such that $\xi$ holds.
  Whenever such an $m$ is found, the machine enters a special state $q_s$ and moves to the next value of $k$.
  The configuration also contains an upper half-plane $H$.
  The formula $F$ checks that if $X_2$ contains the word $w_2$, then $X_3$ contains a $q_s$ within the half-plane $H$.

  Suppose $x \vDash \phi$.
  Given a rectangle $R$, take an $X_1$ that correctly encodes $x|_R$ into a word $w$ and contains some $w_1$.
  There exists an $X_2$ that necessarily contains a $w_2$ (because a witness for $w$ not encoding the pattern does not exist), and wherever we place the upper half-plane in $X_3$, it contains $q_s$.
  Hence the machine in $X_3$ enters $q_s$ an infinite number of times, and $x|_R \in \Le(X)$.

  Suppose then $x \in X$.
  Every valid configuration $X_1$ contains some rectangle $R$ and some words $w$ and $w_1$ above it.
  If $w$ does not encode $x|_R$, we can find a witness and choose it as $X_2$.
  If it does, we can choose $X_2$ to contain the word $w_2$ that corresponds to a suitable $A_2$.
  Then every $X_3$ satisfies $\psi$, since either it does not have the correct form, or the Turing machine enters $q_s$ an infinite number of times.
  Hence $x \vDash \phi$.

  In the second case, we have a subshift $X \subseteq \A^{\Z^2}$ whose language is defined by a $\Sigma^1_2$ formula $\exists_1 A_1 \forall_1 A_2 \exists k \forall m \xi$, where $\xi$ is computable.
  We construct an MSO formula $\phi = \exists X_1 \forall X_2 \exists X_3 \psi$ that defines $X$.
  Here, $X_1$ contains the infinite computation of a Turing machine that simply lists all finite rectangles $R$ and nondeterministically guesses a pattern $P \in \A^R$ and an infinite binary word $w_1$ for each.
  It also selects one cell $\vec w \in \Z^2$ as the ``origin''.
  The configuration $X_2$ selects one of these rectangles, and either selects a position $\vec v \in R$ and contains a witness gadget for $x_{\vec w + \vec v} = P_{\vec v}$, or contains an infinite binary word $w_2$ at the same position as the word $w_1$ of the encoded rectangle.
  The configuration $X_3$ is similar to the $X_3$ of the $\Pi^1_2$ case, except that the Turing machine looks for an $m$ such that $\xi$ does not hold.
  The formula $F$ is almost as before: if $X_2$ contains a $w_2$, then $X_3$ does not contain a $q_s$ within $H$.

  Suppose $x \vDash \phi$, and choose an $X_1$ from $\phi$ with origin $\vec w$.
  Given a rectangle $R$, we can find it encoded in $X_1$ together with a pattern $P$ and a word $w_1$.
  For each $X_2$ that selects the encoding of $R$ and $P$ and a position $\vec v \in R$, we get a witness for $x_{\vec w + \vec v} = R_{\vec v}$.
  For those $X_2$ that contain a $w_2$ instead, we get an $X_3$ whose simulated Turing machine enters $q_s$ a finite number of times.
  Hence $x|_R \in \Le(X)$.
  
  Suppose then $x \in X$, and choose a $X_1$ with origin $\vec 0$ that correctly encodes every rectangular pattern $P$ in $x$, and for each, contains a word $w_1$ corresponding to the set $A_1$ given by the $\Sigma^1_2$ formula of $\Le(X)$.
  Every $X_2$ either contains a witness for the correctness of a single coordinate in some encoded pattern (which are all correct by definition) or a word $w_2$.
  Then the machine simulated in $X_3$ enters $q_s$ only a finite number of times, so we can place the half-plane $H$ above the last occurrence to satisfy $\psi$.
  Hence $x \vDash \phi$.
\end{proof}

\section{Conclusion}

In Section \ref{sec:msosub} we determined the complexity of $C$-SUB for various classes $C$.
FOSUB turned out to be $\Pi^0_4$-complete, while $\ob{\Pi}_n$-SUB and $\ob{\Sigma}_n$-SUB for $n \geq 1$ are both $\Pi^1_n$-complete.

The results of Section~\ref{sec:complexity} are summarized in Table \ref{tab:sumMSOlanguages}.
Note that we have completely characterized the $\ob{\Sigma}_n$ subshifts with $n \geq 2$, as well as the languages of $\ob{\Sigma}_n$ sets.
We have also pinpointed the maximal complexity of the languages of first order and $\ob{\Sigma}_1$ subshifts, which are both $\Sigma^0_2$, higher than the well known $\Sigma^0_1$ bound of SFTs and sofic shifts.
The main remaining question concerns the maximal complexity of $\ob{\Pi}_1$ subshifts: is it $\Sigma^1_n$, the same as for general $\ob{\Pi}_1$ sets, or is it strictly lower?

\begin{table}[htp]
  \centering
  \caption{Summary of the complexity of MSO languages.}
  \label{tab:sumMSOlanguages}
  {
    \setlength{\tabcolsep}{0.5em}
    \renewcommand{\arraystretch}{1.2}
    \begin{tabular}{c|c|c|c|c|}
      & \multicolumn{2}{c}{Sets} & \multicolumn{2}{|c|}{Subshifts} \\
      & Max complexity & Includes all in & Max complexity & Includes all in \\
      FO & $\Sigma^0_2$ & $\Delta^0_1 \not\subseteq$ & $\Sigma^0_2$ & $\Delta^0_1 \not\subseteq$ \\
      $\ob{\Sigma}_1$ & $\Sigma^0_2$ & $\Delta^0_1 \not\subseteq$ & $\Sigma^0_2$ & $\Delta^0_1 \not\subseteq$ \\
      $\ob{\Pi}_1$ & $\Sigma^1_1$ & $\Delta^0_1 \not\subseteq$ & $[\Pi^0_3, \Sigma^1_1]$ & $\Delta^0_1 \not\subseteq$ \\
      $\ob{\Sigma}_n, n \geq 2$ & $\Sigma^1_{n-1}$ & $\Sigma^1_{n-1}$ & $\Sigma^1_{n-1}$ & $\Sigma^1_{n-1}$ \\
      $\ob{\Pi}_n, n \geq 2$ & $\Sigma^1_n$ & $\Pi^1_{n-1}$ & $[\Pi^1_{n-1}, \Sigma^1_n]$ & $\Pi^1_{n-1}$
    \end{tabular}
  }
\end{table}

\begin{credits}
  \subsubsection{\ackname}
  Ilkka Törmä was supported by the Academy of Finland grant 359921.

  \subsubsection{\discintname}
  The authors have no competing interests.
\end{credits}

\bibliographystyle{plainurl}
\bibliography{citations}

\newpage

\appendix
\section{Appendix: Skipped or shortened proofs}

\begin{proof}[of Lemma \ref{equivclass}]
%We know thanks to \cite{rect} that $X_{\phi}$ is indeed a union of equivalence classes and that $(n,k)$ is computable from $\phi$. Let us prove that we can compute this union. 
We may assume that the radius of $\phi$ is at most 1, by introducing new existentially quantified variables for intermediate terms if needed; for example, a subformula $v = \mathrm{East}(\mathrm{East}(w))$ can be replaced by $\exists u \ v = \mathrm{East}(u) \land u = \mathrm{East}(w)$.
We also assume $\phi$ is in prenex normal form, so $\phi = Q_1 v_1 Q_2 v_2 \cdots Q_m v_m \psi$, where each $Q_i$ is a quantifier and $\psi$ is quantifier-free.

For $0 \le \ell \le m$, consider the alphabet $\A_\ell=\A \times \{0,1\}^\ell$.
For $0 \le i \le \ell$, let $\pi_i$ be the natural projection from $\A_\ell$ to $\A$ if $i = 0$ and to $\{0,1\}$ if $i \ge 1$.
Also, for $\ell < m$ let $\rho_\ell$ be the natural projection from $\A_{\ell+1}$ to $\A_\ell$.
We extend these projections to work on patterns and configurations as well.

A configuration $x = (y, b_1, \ldots, b_\ell) \in \A_\ell$ is \emph{consistent} if there is exactly one 1-symbol on each binary layer $b_i$. Such a consistent configuration \emph{satisfies} $\phi$, if the first layer $y$ satisfies $Q_{\ell+1} v_{\ell+1} \cdots Q_m v_m \psi$ when for each $i \in [1, \ell]$, the value of $v_i$ is the position of the 1-symbol on $b_i$.

Let $n \geq 0$ and $0 \le \ell \le m$.
Let $C_f$ be a $\sim_{(n,k)}$ equivalence class on $\A_\ell$.
We say that $C_f$ is \emph{consistent} if the following conditions hold.
\begin{itemize}
    \item For each $\vec{v} \in \nsq$ and each $j \in \llbracket 1, \ell \rrbracket$, there is exactly one pattern $P \in \A_\ell^{n \times n}$ such that $\pi_j(P_{\vec{v}}) = 1$ and $f(P)>0$. In this case, we must have $f(P)=1$.
    \item For each $j \in \llbracket 1, \ell \rrbracket$, the patterns $P \in \A_\ell^{n \times n}$ such that $f(P) > 0$ and $\pi_j(P)$ contains a 1 can be \emph{glued} together, i.e.\ there exists a single pattern $Q$ over $\A_\ell$ in which all of these patterns appear and with only one 1-symbol on each $\{0,1\}$-layer.
\end{itemize}
Notice that if an equivalence class $C_f$ is consistent then all $x \in C_f$ are consistent.

We inductively define unions $\C_\ell$ of equivalence classes on the alphabets $\A_\ell$, starting from $\C_m$ and finishing with $\C_0$.
Let $\C_m=\bigcup_f C_f$ be the union of those consistent $\sim_{(3,2)}$-equivalence classes on the alphabet $\A_m$ such that all $x \in C_f$ satisfy $\phi$. The set $\C_m$ is computable since all variables are interpreted in a fixed position and the radius of $\phi$ is 1.

Let now $0 \le \ell < m$, and suppose $\C_{\ell+1} = \bigcup_f C_f$ has been defined as a union of $(n,k)$-equivalence classes on $\A_{\ell+1}$ for some $n = 2p+1 \geq 3$ odd and $k \geq 2$.
Let $R = [0, n-1]^2$ and $T = [-p, n+p-1]^2$.
We say that a $\sim_{(n,k)}$ equivalence class $C_g = \bigcap_{P \in \A_{\ell+1}^R} g(P)$ is an \emph{extension} of a $\sim_{(4p+1,n^2+k)}$-equivalence class $C_f = \bigcap_{Q \in \A_\ell^T} f(Q)$ if the following conditions hold.
\begin{enumerate}
    \item For each $P \in \A_\ell^R$, we have
      \begin{equation}
        \label{eq:ext}
    \sum_{\substack{Q \in \A_\ell^T \\ Q|_R=P}} f(Q) \ge \sum_{\substack{Q \in \A_{\ell+1}^R \\ \rho_\ell(Q)=P}} g(Q)
    \end{equation}
    and equality must hold when $g(P \times 0^R)<k$.
    \item Let $Q \in \A_{\ell+1}^T$ be such that $\pi_{\ell+1}(Q_{(p,p)}) = 1$ and the $n \times n$ patterns $Q_1$, $Q_2$, $Q_3$ and $Q_4$ at each corner of $Q$ satisfy $g(Q_i) = 1$. Then we require $f(\rho_{\ell}(Q)) \geq 1$.
\end{enumerate}

Similarly, a configuration $y \in \A_{\ell+1}^{\Z^2}$ is an \emph{extension} of $x \in \A_\ell^{\Z^2}$ if both are consistent and $\rho_\ell(y) = x$.

\begin{claim}
\label{cl:extension}
If $y \in C_g$ is an extension of $x \in C_f$, then $C_g$ is an extension of $C_f$.
\end{claim}

\begin{proof}[of Claim]
  Let $P \in \A_\ell^R$.
  Since $y$ is an extension of $x$, we have
  \begin{equation}
    \label{eq:ext2}
    \sum_{\substack{Q \in \A_\ell^T \\ Q|_R=P}} \#_x(Q) = \sum_{\substack{Q \in \A_{\ell+1}^R \\ \rho_\ell(Q)=P}} \#_y(Q).
  \end{equation}
  Also, $\sum_{\substack{Q \in \A_{\ell+1}^R \\ \rho_\ell(Q)=P}} g(Q) \le n^2+k$, since $g(P \times 0^R) \le k$ and $g(Q) \le 1$ for every other $Q \in \A_{\ell+1}^R$ extending $P$.
  If~\eqref{eq:ext} were false, every $Q \in \A_\ell^T$ such that $Q_{|R}=P$ would thus satisfy $f(Q) < n^2+k$, implying $f(Q) = \#_x(Q)$.
  Using~\eqref{eq:ext2} we now have
  \[
    \sum_{\substack{Q \in \A_\ell^T \\ Q|_R=P}} f(Q) = \sum_{\substack{Q \in \A_\ell^T \\ Q|_R=P}} \#_x(Q) = \sum_{\substack{Q \in \A_{\ell+1}^R \\ \rho_\ell(Q)=P}} \#_y(Q) \geq \sum_{\substack{Q \in \A_{\ell+1}^R \\ \rho_\ell(Q)=P}} g(Q),
  \]
  a contradiction.

  Now suppose that $g(P \times 0^R)<k$.
  In this case, $g(Q) = \#_y(Q)$ holds for all $Q \in \A_{\ell+1}^R$ extending $P$, so the wanted equality.

%  The sum of these numbers, which is equal to $\#_x(P)$, is also strictly less than $n^2+k$.
%  Hence $f(Q) = \#_x(Q) < n^2+k$ for all $Q \in \A_\ell^T$ extending $P$.
%  From~\eqref{eq:ext2} we deduce that~\eqref{eq:ext} holds with equality.

  Finally, if $Q \in \A_{\ell+1}^T$ satisfies the hypothesis of the second condition, then such a $Q$ appears in $y$ since $\pi_{\ell+1}(y)$ contains a single 1-symbol.
  So $\rho_\ell(Q)$ appears in $x$, which implies $f(\rho_\ell (Q)) \ge 1$.
\end{proof}

We have two cases depending on the quantifier $Q_{\ell+1}$.
If $Q_{\ell+1} = \exists$, then let $\C_\ell$ be the union of consistent equivalence classes on $\A_\ell$ such that at least one of their consistent extensions is in $\C_{\ell+1}$.
If $Q_{\ell+1} = \forall$, then $\C_\ell$ is the the union of consistent equivalence classes on $\A_\ell$ such that all of their consistent extensions are in $\C_{\ell+1}$.

Now we prove by downward induction on $\ell$ that for each $0 \le \ell \le m$, the union $\C_\ell$ equals the set $X_\ell$ of consistent configurations $x \in \A_\ell^{\Z^2}$ that satisfy $\phi$.
For the initialization step $\ell = m$, it is clear that $X_m = \C_m$.
Let then $\ell < m$, and suppose the claim holds for $\C_{\ell+1}$. Let $x \in \A_\ell^{\Z^2}$.

Suppose first that $Q_\ell = \exists$. If $x \in X_\ell$, then it has a consistent extension $y \in X_{\ell+1}$, which we obtain by choosing the value of $v_\ell$ so that $\phi$ is still satisfied, and putting a 1-symbol at that coordinate. We have $y \in \C_{\ell+1}$ by the induction hypothesis, and so $x$ is in $\C_\ell$ thanks to the Claim.
%Hence $x$ lies in some equivalence class $C_f \in \C_\ell$.

For the other direction, suppose $x \in \C_\ell$, so $x$ lies in some consistent equivalence class $C_f \subseteq \C_\ell$.
By the definition of $\C_\ell$, there exists a consistent equivalence class $C_g \subseteq \C_{\ell+1}$ that is an extension of $C_f$.
Since $C_g$ is consistent, there exists a pattern $Q \in \A_{\ell+1}^T$ that contains occurrences of all patterns $P \in \A_{\ell+1}^R$ such that $g(P) > 0$ and $\pi_{\ell+1}(P)$ contains a 1.
The corners of this pattern satisfy $g(Q_i) = 1$, so item 2 of the definition of extension guarantees $f(\rho_\ell(Q)) \geq 1$; in particular, $\rho_\ell(Q)$ occurs in $x$.
Let $y \in \A_{\ell+1}^{\Z^2}$ be an extension of $x$ where the 1-symbol in $\pi_{\ell+1}(y)$ is at the middle of an occurrence of $\rho_\ell(Q)$.

Now we prove $g(P) = \min(\#_y(P), k)$ for all $P \in \A_{\ell+1}^R$.
We have $y \in C_g \subseteq \C_{\ell+1}$, and thus $y \in X_{\ell+1}$ by the induction hypothesis, which in turn implies $x \in X_{\ell}$.
If $P$ is an extension of $P' \in \A_\ell^R$ with a 1 on the last $\{0,1\}$ layer, then $g(P) = \#_y(P) \leq 1$ by construction.

Suppose then $P = P' \times 0^R$. If $g(P)<k$, then we have equality in~\eqref{eq:ext} and the right hand side is strictly less than $n^2+k$.
Now every $Q' \in \A_{\ell}^T$ such that $Q'|_R=P'$ appears exactly $f(Q')$ times in $x$.
This implies
\[
  \sum_{\substack{Q' \in \A^R_{\ell+1} \\ \rho_{\ell}(Q')=P'}} \#_y(Q') =
  \sum_{\substack{Q' \in \A^T_{\ell} \\ Q'|_R=P'}} \#_x(Q') =
  \sum_{\substack{Q' \in \A^T_{\ell} \\ Q'|_R=P'}} f(Q') =
  \sum_{\substack{Q' \in \A^R_{\ell+1} \\ \rho_{\ell}(Q')=P'}} g(Q').
\]
Then we can conclude that $\#_y(P' \times 0^R)=g(P' \times 0^R)$.

The only one remaining case is when $g(P' \times 0^R)=k$.
%Assume for a contradiction that $\#_y(P' \times 0^R)<k$. Then $\#_y(Q') < k$ for every $Q' \in \A^R_{\ell+1}$ extending $P'$.
Using~\eqref{eq:ext2} and~\eqref{eq:ext} we compute
\begin{align*}
  \sum_{\substack{Q' \in \A^R_{\ell+1} \\ \rho_{\ell}(Q')=P'}} \#_y(Q')= {} &
  \sum_{\substack{Q' \in \A^T_{\ell} \\ Q'|_R=P'}} \#_x(Q') \ge
  \sum_{\substack{Q' \in \A^T_{\ell} \\ Q'|_R=P'}} f(Q') \\
  {} \ge {} &
  \sum_{\substack{Q' \in \A^R_{\ell+1} \\ \rho_{\ell}(Q')=P' \\ Q' \neq P' \times 0^R}} g(Q') =
  k + \sum_{\substack{Q' \in \A^R_{\ell+1} \\ \rho_{\ell}(Q')=P' \\ Q' \neq P' \times 0^R}} \#_y(Q').
\end{align*}
This implies $\#_y(P' \times 0^R) \ge k$, as desired.

The case of $Q_\ell = \forall$ is dual to the above.
If $x \in C_f \subseteq \C_\ell$ and $y \in \A_{\ell+1}^{\Z^2}$ is an extension of $x$, then by the Claim the equivalence class $C_g$ of $y$ is an extension of $C_f$.
This implies $C_g \subseteq \C_{\ell+1}$ by the definition of $\C_\ell$, and by the induction hypothesis $y \in X_{\ell+1}$.
Hence $x \in X_\ell$.

For the other direction, if $x \in X_\ell$, then all extensions of $x$ are in $X_{\ell+1} = \C_{\ell+1}$.
Let $C_f$ be the equivalence class of $x$ and $C_g$ a consistent extension of $C_f$; we wish to show that $C_g \subseteq \C_{\ell+1}$.
As in the $Q_\ell = \exists$ case, let $y \in \A_{\ell+1}^{\Z^2}$ be an extension of $x$ such that the 1-symbol in $\pi_{\ell+1}(y)$ is in the middle of the pattern $\rho_\ell(Q)$ given by item 2 of the definition of extension.
As before, we can prove $y \in C_g$, and since $y \in X_{\ell+1}= \C_{\ell+1}$, this implies $C_g \subseteq \C_{\ell+1}$ and so $x \in C_\ell$.

%Suppose now that $Q_{\ell}=\forall$. If $x \in X_{\ell}$, then all extensions of $x$ are in $X_{\ell+1}$. By the induction hypothesis, they are also in $\C_{\ell+1}$. Let $C_f$ be the equivalence class of $x$ and $C_g$ a consistent extension of $C_f$. Let $y \in \A_{\ell+1}^{\Z^2}$ be the extension of $x$ with a 1-symbol in a middle of a pattern whose existence is given by the point (2) of $C_g$. We can verify thanks to the same technique as before that $y \in C_g$ and since $y \in \C_{\ell+1}$ then $C_g \subseteq \C_{\ell+1}$. Therefore $x \in C_f \subseteq \C_{\ell}$.

%If $x \in C_f \subseteq \C_{\ell}$, then for all $C_g$ consistent extensions of $C_f$, $C_g \subseteq \C_{\ell+1}$. Let $y \in \A_{\ell+1}^{\Z^2}$ which is an extension of $x$. $y \in C_g$ where $C_g$ is a consistent extension of $C_f$ by Claim~\ref{cl:extension}, so $y \in \C_{\ell+1}$ and then in $X_{\ell+1}$ by induction hypothesis. We have $x \in X_{\ell}$.

We conclude that $X_{\phi}=\C_0$.
This construction is computable, since there is only a finite number of $\sim_{(n,k)}$-equivalence classes for each $n$, $k$ and alphabet $\A_\ell$.
\end{proof}

\begin{proof}[of Lemma \ref{mfoin}]
    Let $\phi$ be a first order formula.
  Using Lemma \ref{equivclass} we can compute $n, k \in \N$ and disjoint $\sim_{(n,k)}$-equivalence classes $C_1, \ldots, C_r$ such that $X_{\phi}=\bigcup_{i=1}^r C_i$. We want to prove that knowing whether $\bigcup_{i=1}^r C_i$ is closed is $\Pi^0_4$.

  For a $\sim_{(n,k)}$-equivalence class $C=\iC{\A}$, define the function $f_C : \A^{n^2} \to \{0, \ldots, k\}$ by $f_C(P)=a_P$.
  A \emph{map} of $C$ is a function $M: \A^{n^2} \rightarrow \Pe_{\mathrm{fin}}(\Z^2)$. % such that $|M(P)| \le f_C(P)$ for all $P$.
  For $\ell \in \N$ and two maps $M_1, M_2$ of $C$, we say that $(C,\ell,M_1,M_2)$ \emph{tiles the plane} if there exists $x \in \A^{\Z^2}$ such that for each $P \in \A^{n^2}$:
\begin{itemize}
    \item if $f_C(P) < k$, then the pattern $P$ appears at positions $M_1(P) \cup M_2(P)$ and no other positions,
    \item if $f_C(P) = k$, then the pattern $P$ appears on positions $M_1(P) \cup M_2(P)$ (with $|M_1(P) \cup M_2(P)|\ge k$), and if $|M_1(P)|<k$, the only positions of $[-\ell,\ell-1]^2$ at which $P$ appears are in $M_1(P)$.
    \end{itemize}
    To decide whether a given quadruple tiles the plane is $\Pi^0_1$.
    Namely, by compactness, it is equivalent to the condition that for all $m \geq L$, where $L$ is the maximum of $\ell$ and $\| \vec v \|_\infty + n$ for $P \in \A^{n^2}$ and $\vec v \in M_1(P) \cup M_2(P)$, there exists a finite pattern $Q \in \A^{[-m,m]^2}$ that satisfies the two conditions above.
    The configuration $x$ exists as a limit point of such patterns.
  
  Now let us consider the following arithmetical formula $F$:
  there exist $m \in \N$, an index $i \in \{1, \ldots, r\}$ and two functions $I_1, I_2 : \A^{n^2} \to \{0, \ldots, k\}$ such that $I_1 + I_2 = f_{C_i}$, and
  for all $\ell > m$ there exist two maps $M_1, M_2$ of $C_i$ such that
  \begin{itemize}
  \item $|M_j(P)| = I_j(P)$ for each $j = 1,2$ and $P \in \A^{n^2}$,
  \item $M_1(P) \subseteq [-m,m]^2$ and $M_2(P) \cap [-\ell,\ell-1]^2 =\emptyset$ for each $P \in \A^{n^2}$,
  \item $(C_i,\ell,M_1,M_2)$ tiles the plane, and
  \item there is no $j \in \{1, \ldots, r\}$ such that $(C_j,0,M_1,P \mapsto \emptyset)$ tiles the plane.
  \end{itemize}
  By its form, $F$ is $\Sigma^0_4$.

  We prove that $F$ holds if and only if $\bigcup_{i=1}^r C_i$ is not closed.
  Suppose first that $F$ holds with $m$, $i$, $I_1$ $I_2$ and the maps $M^\ell_1, M^\ell_2$ and tilings $(x_\ell)_{\ell > m}$ they induce. By compactness, we may assume without loss of generality that $x_\ell$ converges to a limit configuration $x \in \A^{\Z^2}$, and that the maps $M^\ell_1 = M_1$ are all equal since their image is contained in $[-m,m]^2$. If $x \in C_j$ for some $j$, then $x$ is a tiling for $(C_j,0,M_1, P \mapsto \emptyset)$, as the positions $\bigcup_{P \in \A^{n^2}} M^\ell_2(P)$ move away from the origin as $\ell$ grows and vanish in the limit. But such a tiling does not exist by the assumption that $F$ holds. Hence we have $x \notin \bigcup_{j=1}^r C_j$, and the set is not closed.

Suppose then that $\bigcup_{i=0}^r C_i$ is not closed. Let $(x_j)_{j \in \N}$ be a sequence of tilings in $\bigcup_{i=0}^r C_i$ such that $x_j \overset{j \rightarrow \infty}{\longrightarrow} x \notin \bigcup_{i=0}^r C_i$. Assume without loss of generality that there is an index $i$ such that $x_j \in C_i$ for all $j$.

Let $m \in \N$ be such that all occurrences of every $n \times n$ pattern which appears less than $k$ times in $x$ are in $[-m,m]^2$. For $P \in \A^{n^2}$, let $I_1(P)$ be the number of occurrences of $P$ in $x$ counted up to $k$, and $I_2(P) = f_{C_i}(P) - I_1(P)$.
%Let $(I_1,I_2)$ be a $C_i$ partition such that $I_1$ gives exactly the patterns in $x$.
Then, given $\ell > m$, we can choose a large $j$ and define $M_1$ and $M_2$ to give the positions of patterns in $x_j$ inside and outside of $[-m,m]^2$, choosing some subsets of positions of sizes $I_1(P)$ and $I_2(P)$ if there are more that that. Doing so, $M_1$ maps only in the square $[-m,m]^2$ and $M_2$ only outside the square $[-\ell,\ell-1]^2$ if $j$ is chosen large enough.

We claim that these choices verify the formula $F$.
The first two items hold by the definition of the $M_j$, and the quadruple $(C_i, \ell, M_1, M_2)$ tiles the plane as witnessed by $x_j$.
Finally, the quadruple $(C_j, 0, M_1, P \mapsto \emptyset)$ does not tile the plane for any $j \in \{1, \ldots, r\}$; if it did, and $y \in \A^{\Z^2}$ was the witness, then necessarily $y \in C_j$ and $y \sim_{(n,k)} x$, implying that $x \in C_j$, a contradiction.
\end{proof}

\begin{proof}[of Lemma \ref{mfohard}]
  We reduce the $\Pi^0_4$-complete problem $\forall COF$: given a Turing Machine $M$ which takes two inputs, does $M(v,\_)$ have a cofinite language for all $v$?
For this, we define an SFT $X$ %(which is definable by a first order formula)
and add on it a first order formula to get a formula $\phi$. Then, we will prove that $M \in \forall COF$ if and only if $\phi \in \MFOSUB$.

Let $M$ be a turing machine. The SFT $X$ is the one described by Figure \ref{fig:mfohard} (the only allowed patterns are those of size $2 \times 2$ in this figure, but other inputs $v$ and $w$ are possible, even the infinite one). %We use standard techniques to simulate $M$ (see Property \ref{computationtiles}).
The input $v$ is copied to the southeast until it hits the gray diagonal signal emitted by the ${*}$-symbol at the end of the second input $w$, where it is erased.
The diagonal turns black after the point of erasure.
%Configurations of $X$ with two finite inputs $x$ and $y$ and a simulated computation are called \emph{classical configurations}.
%but keep in mind that there are many others since it is a subshift.
The machine $M$ is not allowed to halt, so in all configurations containing two finite inputs $v$ and $w$, $M(v,w)$ does not halt. 

Let $\phi'$ be the formula which defines this SFT, and let $\phi=\phi' \land ((\exists x \ \exists y \ P_1(x) \land P_2(y)) \Rightarrow \exists z \ P_3(z))$, where tiles 1, 2 and 3 are showed in Figure \ref{fig:mfohard}). Intuitively, $\phi$ requires the configuration to satisfy $\phi'$, and if there is a diagonal with a finite input $v$ written on it, then there also exists a finite input $w$.

We show that $M \in \forall COF$ if and only if $\phi \in MFOSUB$.
Suppose first $M \notin \forall COF$. We prove that $X_{\phi}$ is not closed. Let $v$ be an input such that for infinitely many inputs $w_1, w_2, w_3, \ldots$ the computation $M(v,w_n)$ does not halt. For each $n$, let $x_n$ be a configuration of $X_{\phi}$ that contains a simulated computation with the inputs $v$ and $w_n$, and the endpoint of the black diagonal line is at the origin. This sequence has a limit configuration $x \in \A^{\Z^2}$, which is not in $X_{\phi}$ since it contains both black and gray diagonal tiles but no $\$$. Hence $X_{\phi}$ is not a subshift.

Now we suppose that $X_{\phi}$ is not a subshift (hence not closed) and prove that $M \notin \forall COF$. Let $(x_n)_{n \in \N}$ be such that $x_n \overset{n \rightarrow \infty}{\longrightarrow} x$ with $x_n \in X_{\phi}$ for all $n$ and $x \notin X_{\phi}$. Notice that the $x_n$ and $x$ are in $X_{\phi'}$ since this set is closed. As $x$ does not satisfy $\phi$, it must have a diagonal with a finite input $v$ on it but no computation tile. For large $n$, all $x_n$ have such a diagonal with the same $v$. But they are in $X_{\phi}$, so they must contain $\$$-symbols, and hence simulated computations. Since $x_n$ converges to $x$, it means that for all $m$, there exists $x_n$ such that the second input $w$ written on it is longer than $n$ symbols. It implies that for infinitely many $w$, $M(v,w)$ does not halt, and so $M \notin \forall COF$.
\end{proof}

% \begin{proof}[of Lemma \ref{pisubin}]
%   Let $\phi \in \Pb{n}$ with $m$ second order quantifiers. Denote by $\phi^*$ the first order formula on $\A \times \{0,1\}^m$ obtained from $\phi$ by removing second order quantifiers and replacing each $v \in X_i$ with $\bigvee_{c \in \pi_i^{-1}(1)} P_c(v)$, where $\pi_i (c_0, \hdots, c_m)=c_i$. For a configuration $x \in \A^{\Z^2}$, we have $x \vDash \phi$ if and only if
%   \begin{equation}
%     \label{eq:pisubin}
%     \forall X_1 \ \hdots \ Q X_m \ x \times X_1 \times \hdots \times X_m \vDash \phi^*
%   \end{equation}
% for $Q=\exists$ if $n$ is even, $Q=\forall$ otherwise.
% From Corollary \ref{fofonda} we deduce that $x \times X_1 \times \hdots \times X_m \vDash \phi^*$ is a $\Delta^0_2$ condition.
% With Remark \ref{calcfonda} it follows that \eqref{eq:pisubin} is $\Pi^0_2$ for $n = 1$, and $\Pi^1_{n-1}$ for $n \geq 2$.

% Now we want to know whether $X_{\phi}$ is closed. It is equivalent with:

% For all $(x_n)_{n \in \N} \in \A^{\Z^2}$ such that

% \begin{itemize}
%     \item for all $n$, for all $m \ge n$,  $x_{n_{|[-n,n]^2}}=x_{m_{|[-n,n]^2}}$. It means intuitively that $(x_n)_{n \in \N}$ converges quickly to $x$ such that $x_{i,j}=(x_{max(|i|,|j|)})_{{i,j}}$. This condition is in $\Pi^0_1$.

%     \item for all $n$, $x_n \vDash \phi$. This condition is in $\Pi^1_{n-1}$.
% \end{itemize}

% then $x \vDash \phi$. This last condition is in $\Pi^1_{n-1}$ so the entire predicate is in $\Pi^1_n$.
% \end{proof}

\begin{proof}[of Lemma \ref{sisubhard}]
Let $n \ge 1$.
To prove that $\Sb{n}$-SUB is $\Pi^1_n$-hard, we reduce an arbitrary problem in $\Pi^1_n$. Let $R$ be a computable predicate. We construct a formula $\phi \in \Sb{n}$ such that
$\forall A_1 \hdots Q A_n \ \ob{Q} k \  Q m \ R$ holds if and only if $X_\phi$ is closed,
with $Q=\exists$ if $n$ is even, $Q=\forall$ otherwise.

Let $\delta$ be the formula which describes a northeast quarter-plane coloured with 0 and 1 such that each row is coloured in a same way, and the other tiles are blank. All limit points of such configurations are also allowed. Let $\phi'$ be the formula for a two-layer tiling where the first layer satisfies $\delta$ and the second layer is coloured with the alphabet $\{0,1\}$, but such that at most one position holds a $1$-symbol. Then $X_{\phi'}$ is a subshift. The infinite binary sequence repeated on each line of a quarter-plane is interpreted as the oracle $A_1$ in the following.

Let $\phi=\exists X_2 \in X_{\delta} \hdots \ Q X_n \in X_{\delta} \ Q \ Z \ \phi' \wedge \psi$, where $\psi$ is the configuration stating the following: if the main configuration $x \in X_{\phi'}$ has a northeast quarter-plane, then
\begin{itemize}
\item each of $X_2, \ldots, X_n$ has one as well at the same position,
\item the configuration $Z$ contains, in the same quarter-plane, a simulated computation of a Turing machine $M$ that can use the binary sequences of $x$ and the $X_i$ as oracles,
\item if $\ob{Q} = \exists$ and the second layer is all-$0$, then $M$ visits a special state $q_s$ a finite number of times (a $\Sb{1}$ condition by Example \ref{ex:finite}), and
\item if $\ob{Q} = \forall$, then the second layer is all-$0$ and $M$ visits the state $q_s$ an infinite number of times.
\end{itemize}
  Notice that the configuration $Z$ is determined by what was quantified before, so its quantifier is irrelevant.
  Hence $\phi$ is $\Sb{n}$.

  The machine $M$ is exactly as in the proof of Lemma \ref{pisubhard}, except that looks for numbers $k, m$ such that $R(A_1, \ldots, A_n, k, m)$ does hold.

Notice that for the configuration $x$ with $A_1$ on the quarter-plane and an all-$0$ second layer,we have $x \vDash \phi$ if and only if $\exists A_2 \hdots \ Q A_n \ \ob{Q} k \ Q m \ R$ holds.
Hence, if $\forall A_1 \hdots Q A_n \ \ob{Q} k \ Q m \ R$ holds, then $X_{\phi}$ is the cartesian product of two subshifts (the second layer being the at-most-one-$1$ subshift), hence a subshift.
On the other hand, if $\forall X_1 \hdots Q X_n \ \ob{Q} k \ Q m \ R$ is false, then there exists a configuration $x \notin X_\phi$ with an all-$0$ second layer. Since all the configurations with the same first layer but with a $1$-symbol somewhere on the second layer are in $X_{\phi}$, it is not a subshift.
\end{proof}

\begin{proof}[of Theorem \ref{pb1language}]
Let $X$ be the subshift in which all the configurations are blank except at most one infinite word, which contains the code $[M]$ of a Turing machine $M$ followed by an infinite sequence of tiles coloured with $\{0,1,|\}$ (plus all limit configurations of these). Let $\psi$ be the first order formula such that $X=X_{\psi}$.

Let $\phi_1$ be a $\Pb{1}$ formula which requires that if $[M]$ is finite, then there are infinitely many $|$-tiles. We interpret this as $[M]$ being followed by an infinite list of natural numbers in binary.

Let $\phi_2$ be a $\Pb{1}$ formula which requires that if $[M]$ is finite, then for each natural number $n$ in the list, $M$ never halts on $n$. This can be done by universally quantifying on a configuration containing a simulated computation of $M(n)$.

Let $\phi_3$ be a $\Pb{1}$ formula which requires that if $[M]$ is finite, then for each natural number $n$, either $n$ occurs in the list, or $M(n) \downarrow$. This can be done by simulating $M(n)$ and simultaneously searching for $n$ in the list, and requiring that this computation halts at some point.

Let $\phi_4$ be a $\Pb{1}$ formula which requires that if $[M]$ is finite, the list of natural numbers written after $M$ is strictly increasing. Once again, we can do this using a single universal quantifier.

Then let $\phi$ be a $\Pb{1}$ formula equivalent with $\psi \land \phi_1 \land \phi_2 \land \phi_3 \land \phi_4$ (just rewriting this formula in the prenex form). % One can prove that $X_{\phi}$ is a subshift with a $\Pi^0_3$-hard language by reducing the language of Turing machines that diverge on infinitely many inputs, which is known to be $\Pi^0_3$-complete.
  
We prove that $X_{\phi}$ is a subshift. Let $(x_n)_{n \in \N} \in X_{\phi}$ such that $x_n \longrightarrow x$. First, $x \in X_{\psi}$ since $X_{\psi}$ is a subshift. If $x \notin X_{\phi_i}$ for some $i = 1, 2, 3, 4$, then $x$ contains a finite code $[M]$ of a Turing machine. For large enough $n$, this $[M]$ is also written on all the $x_n$. Because there is only one configuration up to translation in $X_{\phi}$ containing $[M]$ -- you can neither move, remove nor add new elements of the list -- then the sequence $(x_n)_{n \in \N}$ is eventually constant. Hence $x \in X_{\phi}$ and $X_{\phi}$ is a subshift.

Now, let us prove that $X_{\phi}$ has a $\Pi^0_3$-hard language. To prove it, we reduce the $\Pi^0_3$-complete problem $coCOF$, which consists of the codes $[M]$ of Turing machines $M$ such that for infinitely many inputs $n$, $M$ does not halt on $n$. In fact we claim that $[M] \in coCOF$ if and only if $[M] \in \Le(X_{\phi})$.

If $[M] \in coCOF$, then the configuration with $[M]$ written at the origin and followed by the infinite list of inputs $n$ on which $M$ does not halt, separated by $|$-tiles and ordered increasingly, is in $X_{\phi}$. Hence $[M] \in \Le(X_{\phi})$.

If $[M] \in \Le(X_{\phi})$, then $[M]$ appears in a configuration $x \in X_{\phi}$ and this configurations gives an infinite list of inputs on which $M$ does not halt. Hence $M \in coCOF$.
\end{proof}

\begin{proof}[of Theorem \ref{Sbncompleteness}]
%We proceed somewhat analogously to the proof of Theorem~\ref{thm:Pbn-definable}.
Let $X \subset \A^{\Z^2}$ be a subshift with a $\Pi^1_{n-1}$ or $\Sigma^1_{n-1}$ language.
The language can be defined by a formula of the type
\[F = Q_1 A_1 \ob{Q}_1 A_2 \cdots Q_1 A_{n-1} \ob{Q} m Q k \xi\]
if $n$ is even, and
\[F = Q_1 A_1 \ob{Q}_1 A_2 \cdots \ob{Q}_1 A_{n-1} Q m \ob{Q} k \xi\]
when $n$ is odd, where $\xi$ is computable.

We construct a $\Pb{n}$ or $\Sb{n}$ formula $\phi$ with $X_\phi = X$, which will have the form $\phi = Q_1 X_1 \ob{Q}_1 X_2 \cdots \ob{Q}_n X_n \psi$ or $\phi = Q_1 X_1 \ob{Q}_1 X_2 \cdots Q_n X_n \psi$ for a first order formula $\psi$, depending on the parity of $n$.
The idea is that the configurations $X_1, \ldots, X_{n-1}$ will contain infinite binary words corresponding to $A_1, \ldots, A_{n-1}$, and $X_n$ will contain geometric and computational structures corresponding to $m$ and $k$.
In addition, $X_1$ and $X_2$ are used to extract the contents of the main configuration $x \in \A^{\Z^2}$ into a form that is usable by a simulated Turing machine.
In the $\Pi^1_{n-1}$ case we use the technique of \cite[Theorem 3]{torma} to perform this extraction, while in the $\Sigma^1_{n-1}$ case we employ a new method.
In both cases $X_1$ encodes rectangular patterns over $\A$, and $X_2$ has one of two forms: type-1 configurations are used to verify that the patterns of $X_1$ correspond to patterns in the main configuration, and type-2 configurations analyze these patterns together with the remaining configurations $X_3, \ldots, X_n$.

Suppose first we are in the $\Pi^1_{n-1}$ case.
The universally quantified $X_1$ contains a finite rectangle $R$.
The row directly above its north border is highlighted.
On the highlighted row is a finite word $w \in \A^*$ that may or may not encode the contents of $R$ in the main configuration $x \in \A^{\Z^2}$, followed by an infinite binary word $w_1 \in \{0,1\}^\N$ that will correspond to the set $A_1$ in formula $F$.
From the southeast corner of the rectangle, a signal travels to the northeast until it hits the highlightred row.
%If $X_1$ does not have this structure, we say it is \emph{malformed}.

The second, existentially quantified configuration $X_2$ has one of two types.
A type-1 configuration contains a single $\$$-symbol inside $R$, from which two signals emanate, one to the north and another to the northeast.
The former signal carries information about the contents of the $\$$-cell in the main configuration $x$.
The signals terminate when they reach the highlighted row.
See Figure~\ref{fig:Pbncompleteness}.

\begin{figure}[htp]
  \centering
  \begin{tikzpicture}
    \draw [fill=black!20] (0,1) rectangle (5,3);
    \draw (0,3) rectangle (9,3.5);
    \draw [dashed,->] (5,1) -- (7.25,3.25);
    \draw (9,3) -- ++(1,0);
    \draw (9,3.5) -- ++(1,0);
    \draw [dashed] (10,3) -- ++(1,0);
    \draw [dashed] (10,3.5) -- ++(1,0);
    \draw (0,3.5) -- ++(0,1);
    \draw [dashed] (0,4.5) -- ++(0,1);

    \node [draw,fill=white] (dollar) at (2,2) {$\$$};
    \draw [dashed,->] (dollar) -- (2,3.25);
    \draw [dashed,->] (dollar) -- (3.25,3.25);

    \node at (0.5,2.5) {$R$};
    \node at (0.5,3.25) {$w$};
    \node at (9.5,3.25) {$w_1$};

    \draw [dotted,->] (0,3.5) -- ++(1,1) -- ++(-0.5,0.5) -- ++(0.5,0.5);
    \node [right] at (1.5,4.5) {\large{Computation of $M$}};
  \end{tikzpicture}
  \caption{Probing a single symbol from $R$ with a type-1 configuration in the $\Pb{n}$-case of the proof of Theorem~\ref{Sbncompleteness}.}
  \label{fig:Pbncompleteness}
\end{figure}
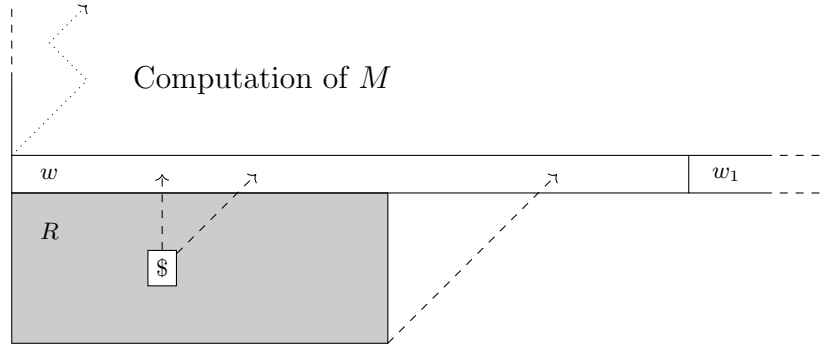

Above the highlighted row we simulate a Turing machine that has as its input the word $w$ and each of the signals on $X_1$ and $X_2$ that terminate on it.
From these signals it is possible to compute the width and height of the rectangle $R$, the relative position of the $\$$-symbol in $R$, and the symbol $a \in \A$ at that position in the configuration $x$.
The machine checks that the word $w$ does not encode a rectangular pattern of the same dimensions as $R$ that contains symbol $a$ at the given position, and then enters a halting state.

A type-2 configuration $X_2$ only highlights the same row as $X_1$.
If $n > 2$, it also contains an infinite binary word $w_2 \in \{0,1\}^\N$ on that row.

Suppose then we are in the $\Sigma^1_{n-1}$ case.
The configuration $X_1$ contains a single symbol $\#$, which we think of as the ``origin''.
The origin is the corner of an infinite northeast quarter-plane inside which we simulate a nondeterministic Turing machine.
The machine runs for an infinite time, enumerating all finite squares $[a,b] \times [c,d] \subset \Z^2$, and for each it nondeterministically guesses a pattern $P \in \A^{[a,b] \times [c,d]}$, writes it on the tape together with the numbers $a, b, c, d \in \Z$, and enters a special state $q_s$ on the left end of its tape.
When it enters state $q_s$, that row also contains an infinite word in $\{0,1\}^\N$ on a separate track.
Hence, the rows of the quarter-plane that start with $q_s$ have the form $q_s u v w_1$, where $u \in \{0,1\}^*$ encodes a square $[a,b] \times [c,d]$, $v \in \A^{(c-a+1)(d-b+1)}$ encodes a pattern $P$ of that shape, and $w_1 \in \{0,1\}^\N$ is an infinite binary word that will correspond to the set $A_1$ in the formula $F$.

The configuration $X_2$ can have one of several types.
A type-1 configuration highlights a single row of the quarter-plane of $X_1$, which must contain a $q_s$, and hence have the form $q_s u v w_1$ as above.
This row splits the right half-plane containing it into two parts, which are colored differently.
Hence the configuration has four parts with different colors: the highlighted row, the quarter-plane $Q_1$ above the row, the quarter-plane $Q_2$ below the row, and the remaining half-plane $H$.
The configuration also contains exactly one $\$$-symbol, which can \emph{a priori} be placed anywhere.
Within $X_2$, a signal is emitted from the $\#$-symbol in $X_1$ to the northeast, and at most two signals are emitted from the $\$$ in $X_2$ whose directions depend on whether it lies in the highlighted row (no signals), $Q_1$ (south and southeast), $Q_2$ (north and northeast) or $H$ (east and southeast, each turning southeast or northeast toward the highlighted row when meeting $Q_1$ or $Q_2$).
See Figure~\ref{fig:Sbncompleteness}.
Each of these signals remembers its source, and in the case of $\$$, also the symbol under it in the main configuration.

\begin{figure}
\centering
\begin{tikzpicture}
  
  	\draw (0,-1) -- (0,5);
  	\draw (0,3) -- (6,3);
  	
  	\draw [dashed, -{Straight Barb[length=1.5mm]}] (1,1) -- (3,3) (0,0) -- (1,1);
  	\node [draw,fill=white] at (0,0) {$\#$};
  	\node [draw,fill=white] at (0,3) {$q_s$};
  	\draw [dashed, -{Straight Barb[length=1.5mm]}] (4,4) -- (4,3);
  	\draw [dashed, -{Straight Barb[length=1.5mm]}] (4.5,3.5) -- (5,3) (4,4) -- (4.5,3.5);
  	\node [draw,fill=white] at (4,4) {$\$$};
  	\draw [dashed, -{Straight Barb[length=1.5mm]}] (2.5,1) -- (2.5,3);
  	\draw [dashed, -{Straight Barb[length=1.5mm]}] (3.5,2) -- (4.5,3) (2.5,1) -- (3.5,2);
  	\node [draw,fill=white] at (2.5,1) {$\$$};
  	\draw [dashed, -{Straight Barb[length=1.5mm]}] (-2.5,1.5) -- (0,1.5);
  	\draw [dashed, -{Straight Barb[length=1.5mm]}] (0.5,2) -- (1.5,3) (0,1.5) -- (0.5,2);
  	\draw [dashed, -{Straight Barb[length=1.5mm]}] (-1.5,2.5) -- (0,4) (-2.5,1.5) -- (-1.5,2.5);
  	\draw [dashed, -{Straight Barb[length=1.5mm]}] (0.5,3.5) -- (1,3) (0,4) -- (0.5,3.5);
  	\node [draw,fill=white] at (-2.5,1.5) {$\$$};
  	
  	\foreach \x in {3,4,4.5,2.5,5,1.5,1}{
  		\fill (\x,3) circle (0.05cm);
  	}
  	
  	\node at (3,-0.5) {$Q_1$};
  	\node at (3,4.5) {$Q_2$};
  	\node at (-2,4) {$H$};
  	
  \end{tikzpicture}
  \caption{The signals of a type-1 configuration in the $\Sb{n}$-case of the proof of Theorem~\ref{Sbncompleteness}. Several choices for the $\$$-symbol are shown; an actual configuration would contain only one.}
  \label{fig:Sbncompleteness}
  \end{figure}

In the quarter-plane $Q_1$ we simulate a Turing machine whose input is the highlighted row, including any signals that pass it.
From these signals the machine can determine the relative positions of the $\#$-symbol and the $\$$-symbol -- let $\vec v \in \Z^2$ be their difference -- plus the symbol $s \in \A$ of the main configuration under the $\$$.
The machine checks that if $\vec v \in [a,b] \times [c,d]$, then $P_{\vec v} = s$, and enters a halting state if the check succeeds.

A type-2 configuration highlights a row $r$ of the quarter-plane in $X_1$ containing a $q_s$.
If $n > 2$, it also contains an infinite binary word $w_2 \in \{0,1\}^\N$ on that row.

The remaining part of the construction is identical between the $\Pi^1_{n-1}$ and $\Sigma^1_{n-1}$ cases.
If the configuration $X_2$ is type-2, then the configurations $X_3, \ldots X_{n-1}$ all contain their own infinite binary words $w_3, \ldots, w_{n-1} \in \{0,1\}^\N$ on the same row $r$ as $w_2$, and $X_n$ contains a second highlighted row $r'$ somewhere to the north of $r$.
Starting from $r$, it also contains a simulation of a Turing machine that reads the pattern $P \in \A^{[a,b] \times [c,d]}$ encoded on row $r$ in $X_1$.
%If $n$ is even,
In the case that the formula $F$ ends in $\forall m \exists k \xi$, using the infinite binary words $w_1, \ldots, w_{n-1}$ as the sets $A_1, \ldots, A_{n-1}$, the machine checks for increasing $m$ whether there exists $k$ such that $\xi(P, A_1, \ldots, A_{n-1}, m,k)$ holds, only moving from $m$ to $m+1$ when such a $k$ is found.
Whenever it finds a $k$ for a new $m$, it enters a special state $q'_s$.
% If $n$ is odd,
If  $F$ ends in $\exists m \forall k \xi$ instead, it behaves similarly, except that it looks for a $k$ for which $\xi(P, A_1, \ldots, A_{n-1}, m,k)$ does not hold.
Note that if $n = 2$, then we have $X_2 = X_n$, so this configuration contains two highlighted rows.
In the case that $X_2$ is type-1, we put no constraints on the configurations $X_3, \ldots, X_{n-1}$.

We define the formula $\psi$ as follows: either $X_2$ contains a $\$$ (and hence is of type 1) and a halting state of the simulated Turing machine, or $X_2$ does not contain a $\$$ (and hence is of type 2) and $X_n$ contains a $q'_s$ above the northmost highlighted row if and only if $F$ ends in $\forall m \exists k \xi$.
The configurations $X_1, X_2, \ldots, X_n$ are required to have the forms described above as they are quantified (see the remark above Example \ref{ex:finite}).

We claim that $X_\phi = X$.
Let us handle the $\Pi^1_{n-1}$ case.
Suppose first that $x \in X$, and let $X_1$ be chosen arbitrarily.
It contains a rectangle $R$, and above it, the finite word $w$ and the infinite word $w_1$.
If $w$ does not correctly encode the pattern $x|_R$, then either $R$ is not of the correct shape, in which case we can choose $X_2$ to contain a computation that verifies this, or $x|_R$ and the pattern encoded by $w$ differ in at least one position $\vec v$, and then we can place the $\$$-symbol at that position and let the simulated machine detect the discrepancy.
In both cases we have $x \vDash \phi$ regardless of the remaining $X_i$-configurations.

If $w$ does encode $x|_R$, then we choose $X_2$ as a type-2 configuration.
The words $w_2, w_3, \ldots, w_{n-1}$ can now be chosen or given according to the quantifiers of the $X_i$ in such a way that, depending on the parity of $n$, either $\forall m \exists k \ \xi(x|_R, w_1, \ldots, w_{n-1}, m, k)$ or $\forall m \exists k \ \xi(x|_R, w_1, \ldots, w_{n-1}, m, k)$ holds, since $x|_R \in \Le(X)$.
Then the Turing machine simulated in any $X_n$ enters state $q'_s$ infinitely many times for odd $n$, and finitely many times for even $n$.
Hence, we can choose the northmost highlighted row in $X_n$ arbitrarily for odd $n$, and above every occurrence of $q'_s$ for even $n$.
This shows $x \in X_\phi$.

Suppose then that $x \vDash \phi$, and let $R \subset \Z^2$ be a finite rectangle.
We show that $P = x|_R \in \Le(X)$.
Let $X_1$ contain the rectangle $R$ and above it a word $w$ that encodes $P$, followed by an arbitrary $w_1 \in \{0,1\}^\N$.
Then $X_2$ cannot have type 1, so it must have type 2.
If $n > 2$, the remaining configurations $X_3, \ldots, X_{n-1}$ and words $w_2, \ldots, w_{n-1}$ are chosen to encode the corresponding sets $A_i$ given by $F$ if existentially quantified, or chosen arbitrarily if universally quantified.
After these configurations are fixed, the Turing machine simulated in $X_n$ must enter state $q'_s$ infinitely often when $n$ is even (since the northmost highlighted row in $X_n$ can be arbitrarily high and the computation is deterministic) and finitely often when $n$ is odd (since it cannot do so above the highlighted row).
But this means that $F$ holds for the pattern $P$, which is thus in the language of $X$.
Since $R$ was an arbitrary rectangle and $X$ is a subshift, this implies $x \in X$.

Now we handle the $\Sigma^1_{n-1}$ case.
Suppose first $x \in X$, and choose the configuration $X_1$ so that the $\#$-symbol is at $(0,0)$, the nondeterministic machine chooses each pattern $P \in \A^{[a,b] \times [c,d]}$ as $x|_P$, and the associated word $w_1 \in \{0,1\}^\N$ is such that $F$ holds for $P$ with $A_1$ as $w_1$, which is possible since $P$ is in the language of $X$.
Choose the configuration $X_2$ arbitrarily.
If it contains a $\$$-symbol at some position $\vec v \in \Z^2$, then it is type-1 and has a highlighted row containing a pattern $P \in \A^{[a,b] \times [c,d]}$.
We either have $\vec v \notin [a,b] \times [c,d]$ or $x_{\vec v} = P_{\vec v}$ due to how we chose the pattern $P$ in $X_1$.
If $X_2$ is type-2, we choose the configurations $X_3, \ldots, X_{n-1}$ either arbitrarily or with the words $w_3, \ldots w_{n-1}$ being such that $F$ holds, depending on their quantifiers.
As in the $\Pi^1_{n-1}$ case, the configuration $X_n$ contains the appropriate number of $q'_s$-states, implying $x \vDash \phi$.

Suppose then $x \in X_\phi$, and let $R = [a,b] \times [c,d]$ be an arbitrary rectangle.
We show that $P = x|_R$ is in the language of $X$.
Let $X_1$ be a configuration given for $x$ by $\phi$, which has the $\#$-symbol at some position $\vec v \in \Z^2$ together with the quarter-plane and the simulated Turing machine.
On some row $r$ the machine has entered state $q_s$ with the rectangle $R - \vec v$ and some pattern $P \in \A^{R - \vec v}$.
For each $\vec w \in R$, there is a type-1 configuration $X_2$ whose $\$$-symbol is at $\vec w$.
It satisfies $\vec w - \vec v \in R$, so we must have $P_{\vec w - \vec v} = x_{\vec v}$.
On the other hand, there is a type-2 configuration $X_2$ whose highlighted row is $r$.
When we choose the remaining configurations $X_3, \ldots, X_n$ according to their quantifiers, the computation simulated in the final configuration $X_n$ again guarantees that, depending on the parity of $n$, either $\forall m \exists k \ \xi(P, w_1, \ldots, w_{n-1}, m, k)$ or $\forall m \exists k \ \xi(P, w_1, \ldots, w_{n-1}, m, k)$ holds.
Hence $P = x|_R \in \Le(X)$, and since $R$ was arbitrary and $X$ is a subshift, we have $x \in X$.
\end{proof}

\end{document}